\documentclass[11pt,a4paper]{amsart}

\usepackage{a4wide}
\usepackage{amssymb}
\usepackage{amsmath}
\usepackage{amsthm}
\usepackage[usenames,dvipsnames,svgnames,table]{xcolor}
\usepackage{graphicx}
\usepackage{hyperref}
\usepackage[utf8]{inputenc}
\usepackage[english]{babel}
\usepackage{xcolor}
\usepackage{bm}

\def\p{{\partial}}
\def\Re{\mathop{\text{Re}}}
\def\Im{\mathop{\text{Im}}}
\def\dist{\mathop{\text{dist}}}

\newcommand{\Cov}{\mathrm{Cov}}

\RequirePackage{color}

\newtheorem{prop}{Proposition}
\newtheorem{lem}{Lemma}
\newtheorem{cor}{Corollary}

\theoremstyle{remark}
\newtheorem{remark}{Remark}

\title{Martingales of stochastic Laplacian growth}
\author[O. Alekseev]{Oleg Alekseev}
\email{teknoanarchy@gmail.com}
\address{
Chebyshev Laboratory, Department of Mathematics and Computer Science, Saint-Petersburg State University, 14th Line, 29b, 199178, Saint-Petersburg, Russia
}
\begin{document}

\begin{abstract}
A family of exponential martingales of a stochastic Laplacian growth problem is proposed. Stochastic Laplacian growth describes a regularized interface dynamics in a two-fluid system, where the viscous fluid is incompressible at a large scale, while compressible at a small scale in the vicinity of the interface. Hence, random fluctuations of pressure near the boundary are inevitable. By using Loewner-Kufarev equation, we study interface dynamics generated by nonlocal random Loewner measure, which produces the patterns with viscous fingers. We use a Schottky double construction to introduce a one-parametric family of functions of random processes on the double closely connected to the correlation functions of primary operators of the boundary conformal field theory in the Coulomb gas framework. For a specific value of the parameter, these functions are martingales with respect to stochastic Loewner flow on the Schottky double. A connection between the proposed algebraic construction and the physical problem of stochastic interface dynamics relies on the Hadamard's variational formula. Namely, the variation of pressure in stochastic Laplacian growth near the interface is given by the covariance of martingales on the double.
\end{abstract}
\maketitle

\section{Introduction}

Loewner approach of studying increasing families of plane compact sets is an important tool in complex analysis since 1920's. This approach allows one to parametrize families of conformal mappings of references domain, e.g., the upper half plane or the unit disk, in terms of real-valued functions on the real line or the unit circle. The Schramm-Loewner evolution (SLE) proposed by Oded Schramm in the early 2000's generalizes the original Loewner method to the case of growing random fractal curves in the plane domain that satisfy conformal invariance and domain Markov property~\cite{Sch00}. From the physical point of view SLE curves describe interfaces in the continuum limits of various lattice models of statistical mechanics, such as the boundaries of spin clusters in the Ising and Potts models, percolation paths, etc.

A more general family of fractals are generated in nonlocal processes, where the growth takes place at multiple points of the interface simultaneously. These processes can be studied by the generalized version of Loewner equation, called Loewner-Kufarev equation~\cite{Kufarev}. In this equation, the driving function is replaced by the driving measure on the real line or the unit circle. An important example of nonlocal deterministic growth process is Laplacian growth (LG) or the Hele-Shaw problem (see Ref.~\cite{BensimonRMP} for a review). This diffusion-driven growth process embraces numerous free boundary dynamics including bidirectional solidification, dendrite formation, electrodeposition, dielectric breakdown, bacterial growth, and flows in porous media~\cite{PelceBook}. LG stands out against the other nonlocal growth processes, because of its powerful properties, unusual for most nonlinear growth processes, such as infinitely many conservation laws~\cite{Ric72} and closed form exact solutions~\cite{SB84}. A new splash of interest to LG was provoked by establishing strong connections of the interface dynamics to major integrable hierarchies, and the theory of random matrices~\cite{MWZ00,KKMWWZ01}.

LG is characterized by the finger-like unstable patterns featuring the formation of cusps at the boundary in a finite time~\cite{SB84}. Hence, the problem is ill defined, and a certain regularization is required. The conventional hydrodynamical regularization is realized through a surface tension, when pressure along the interface is proportional to its curvature. The universal fractal structures are observed in the zero-surface-tension limit. However, surface tension is a singular perturbation of the model, i.e., it has a profound effect upon the evolution of the interface. Hence, the zero-surface-tension limit of the LG problem is rather tricky. Fortunately, there exist alternative methods for regularizing the problem. Probably, the best known are various random walk models of the Hele-Shaw behavior.

In this paper we consider stochastic LG, which describes nonsingular diffusion driven stochastic interface dynamics in the Hele-Shaw cell. A discrete model, which underlies this process, is based on the random walk model of the LG problem, namely, diffusion-limited aggregation (DLA)~\cite{WS81}. DLA is a process where equal particles with a small size $\hbar$ are issued one by one from a distant source. These particles diffuse as random walkers until hit the domain and stick to it. Numerical simulations of DLA show the formation of branching graphs with a width controlled by the size of the particle $\hbar$. The typical aggregates take ramified shapes, and develop branches and fjords of various scales. Remarkably, however, that all grown clusters appear to be monofractals with the numerically obtained Hausdorff dimension $D_H=1.71\pm0.01$, which appears to be robust and universal~\cite{Hal00}. DLA is believed to provide a discrete variant of LG, where the particle size, $\hbar$, plays a role of a short-distance cutoff regularizing singularities emerging in fluid dynamics. The relation of stochastic (DLA) and deterministic (LG) processes remains unclear (see, however, Refs.~\cite{Levitov98,Makarov01}). A natural way to proceed is to consider the continuum limit of stochastic growth, which appears to be non-trivial. It can be anticipated that the DLA dynamics is equivalent to Laplacian contour dynamics, where the finite-time singularities are resolved on microscale due to noise. 

Recently, it was proposed, that LG and DLA can be unified as two opposite limits of a discrete stochastic LG model~\cite{AMW16}. In this model the distant source emits $K\geq1$ uncorrelated particles per time unit. However, no motion of the interface occurs until all of them hit the boundary. The advance of the interface along the unit normal vector is then proportional to the number of particles, which hit the given portion of the boundary. Although the continuum limit of this discrete process is not rigorously obtained yet, a physical intuition based on statistical mechanics arguments allows one to introduce a continuum analog of this model~\cite{QLG1,Ale19a,Ale19b}.

In this work we continue to study stochastic LG in the continuum framework. The model is described by Loewner-Kufarev equation with a random Loewner measure. It possesses a family of martingales~\cite{Ale19b} closely connected to certain correlation functions of conformal field theory (CFT)~\cite{BPZ,DiFrancescoBook}. Let us briefly recall the basic arguments. One can couple various lattice models of statistical mechanics (described by CFTs at critical points in the continuum limit) to any planar domains. The Loewner chain generates a sequence of conformal transformations of the domain, which act in the Hilbert space of CFTs by the Virasoro algebra spanned by the generators $L_n$, $n\in\mathbb Z$, with the commutation relations $[L_n,L_m]=(n-m)L_{n+m}+(c/12)(m^3-m)\delta_{m+n,0}$, where the real constant $c$ is the central charge. The Verma modules $V_{c,h}$ with the highest weights $h$, form the highest weight representation  of the algebra. The Verma module, $V_{c,h}$, can be reducible for certain values of $c$ and $h$, because it may contain the so-called null vectors, i.e., the states in the module, $|\chi\rangle\in V_{c,h}$, that are both primary and descendant, i.e., $L_n|\chi\rangle=0$ for $n>0$. The relevant example is the null vector at the second level. Requiring the decoupling of the null vector from every correlation function leads to the second order differential equation for the correlation functions containing the highest weight vector of the degenerate module.

Remarkably, the same second order differential equation also arises when considering stochastic dynamics of interfaces in the plane. The best known example is SLE, where this equation appears as the necessary condition for certain functions of random processes to be martingales. Therefore, it is natural to expect, that the null vector constraint for CFT coupled to the LG problem can be also related to certain martingales of the interface dynamics.

Below, we will try to avoid any application of CFT technique, and study stochastic LG by the conventional methods of stochastic calculus. The only nontrivial algebraic construction we used below is a Schottky double of the domain complementary to the growing cluster. Note, that this construction has already been appeared in the context of deterministic LG problem~\cite{Gustafsson83,Krichever04}. One can show, that the interface dynamics in the LG problem can be considered as an evolution of a certain curve on the Schottky double. As will be detailed shortly, the martingales are random processes on the double associated with this stochastic dynamics. A direct connection between the proposed martingales and stochastic LG problem relies on Hadamard's variational formula for the Green's function of the Dirichlet boundary problem. We will show, that the variation of pressure near the interface can be written in terms of the covariance of the martingales on the Schottky double.

The article is organized as follows. Section~\ref{s:slg} recalls necessary background on Loewner chains, and both deterministic and stochastic LG problems of simply connected planar domains. We also recall a relation between LG and the Dirichlet boundary problem. This helps to introduce Hadamard's formula for the variation of the Dirichlet Green's function under smooth deformations of the domain in the stochastic LG problem. Section~\ref{s:martingales} is devoted to martingales of stochastic LG. We consider the Schottky double of the domain complementary to the growing cluster, and define a Loewner flow on the double by continuously continuing the flow from the one side of the double to another across the interface. Then, we introduce a one-parametric family of exponential functions of random processes on the Schottky double, and study their variation with respect to the stochastic flow. For specific value of the parameter, these functions are shown to be the martingales of stochastic Loewner flow. Finally, in Section~\ref{s:conclusion} we draw our conclusion and discuss some open problems.

\section{Laplacian growth}\label{s:slg}

\subsection{Loewner chains.} Let us consider a continuously \textit{decreasing} set of simply connected planar domains, $D_t$, $0\leq t<\infty$, and refer to the continuous label $t$ as to time. By the Riemann mapping theorem there exist conformal maps $z_t: \mathbb D\to D_t$ from the exterior of the unit disk, $\mathbb D$, in the auxiliary $w$ plane to the domains $D_t$ in the physical $z$ plane (see Fig.~\ref{map}). In order to make the map unique, one imposes the constraints: $z_t(\infty)=\infty$, and $z_t'(\infty)=r_t$ is a real-valued positive function of time called conformal radius~\footnote{Here and below dot and prime denote the partial derivatives with respect to time and coordinate respectively.}. A family of functions $(z_t)_{t\geq0}$ is called a Loewner chain. The Loewner equation describes the evolution of the Loewner chain $(z_t)_{t\geq0}$ with time by the following partial integro-differential equation:
\begin{equation}\label{z_eq}
	\frac{\p z_t(w)}{\p t}=-wp_t(w)\frac{\p z_t(w)}{\p w},
\end{equation}
where the function $p_t(w)$ is measurable for $t\geq0$ for all $w\in\mathbb D$, holomorphic in $w\in \mathbb D$, and has the negative real part, $\Re p_t(w)<0$ for $w\in\p\mathbb D$. It can be written as follows:
\begin{equation}\label{p-def}
	p_t(w)=\int_0^{2\pi}\frac{d\phi}{2\pi}\frac{e^{i\phi}+w}{e^{i\phi}-w}\rho_t(e^{i\phi}),
\end{equation}
where the Loewner density, $\rho_t(w)$, is a real-valued positive function on the unit circle, which drives the time evolution of the chain. Given the initial condition, $z_0(w)=w$, the Loewner-Kufarev equation admits an unique solution.  Note, that the function $p_t(w)$ has a jump across the unit circle:
\begin{equation}\label{p-limits}
	p_t(w)=\Re p_t(w)+i\Im p_t(w),\qquad p_t(1/\bar w)=-\Re p_t(w)+i\Im p_t(w),
\end{equation}
provided $w=\lim_{\epsilon\to 0}(1+\epsilon)e^{i \theta}$.

\begin{figure}[t]
\centering
\includegraphics[width=.6\columnwidth]{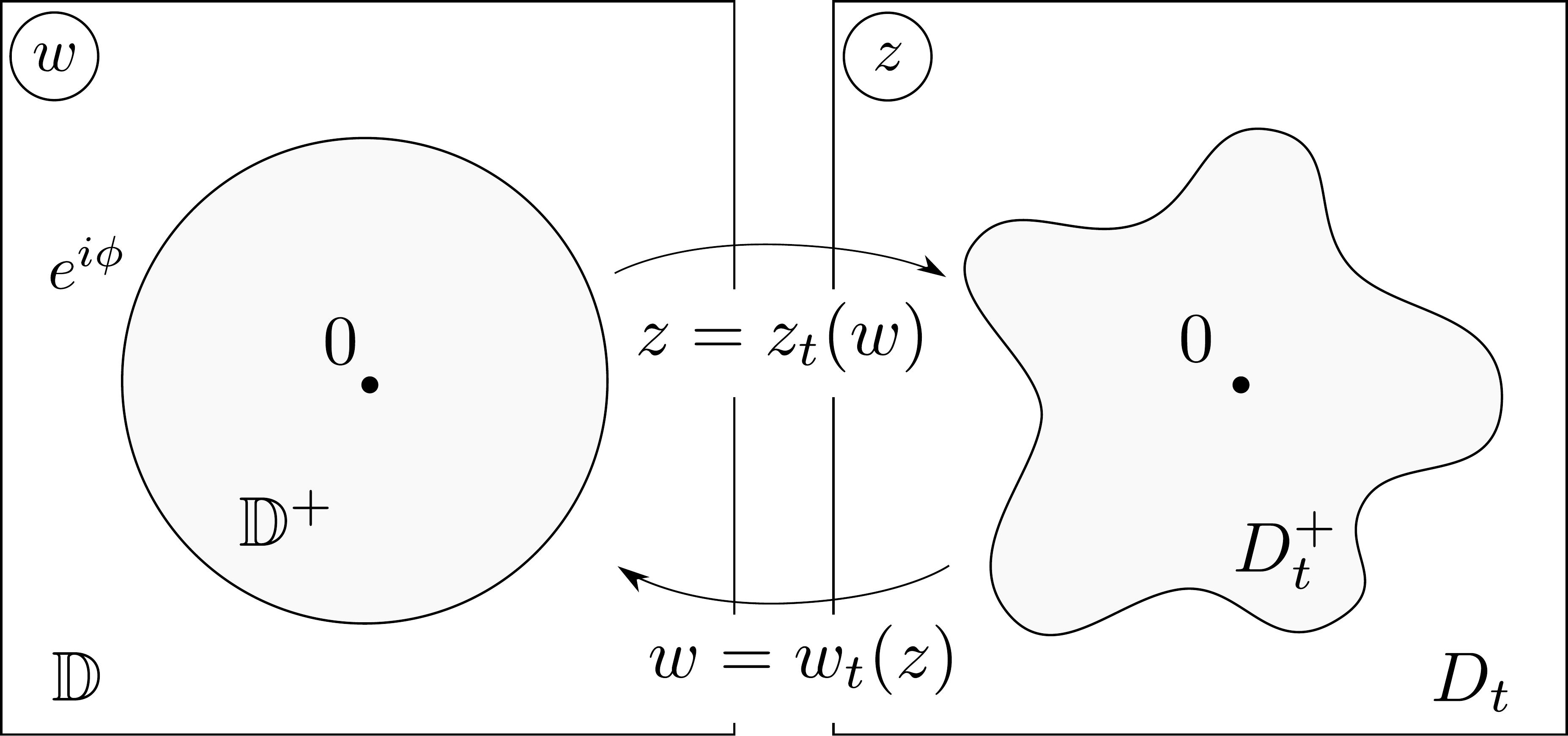}
\caption{\label{map}
The conformal map $z=z_t(w)$ from the exterior of the unit disk $\mathbb D$ in the auxiliary $w$ plane to the exterior of the growing cluster $D_t$ in the $z$ plane. The map is normalized as follows: $z_t(\infty)=\infty$, and $z'_t(\infty)=r_t$, where the conformal radius, $r_t$, is a real-valued function of time $t$.
}
\end{figure}

The Loewner density determines the normal velocity of the contour $\p D_t$ with respect to the time evolution generated by Loewner-Kufarev equation~\eqref{z_eq}. Indeed, the unit complex tangent vector to the contour is $\tau=i w z'_t(w)/|z'_t(w)|$, and the outer normal unit vector is $n=-i \tau$. Thus, by using Loewner equation~\eqref{z_eq}, and taking into account that $\Re p_t(u)=-\rho_t(u)$, $u\in \p\mathbb D$, the normal interface velocity, $v_n(z)=\Re[\bar n\p_t z_t(w)]$, at $z=z_t(e^{i\phi})\in \p D_t$ reads
\begin{equation}\label{v-rho}
	v_n(z_t(e^{i\phi}))=|z'_t(e^{i\phi})|\rho_t(e^{i\phi}).
\end{equation}

Let us briefly mention some particular examples of Loewner chains studied in the literature. The most known example is Loewner evolution generated by a singular Dirac measure, $\rho(e^{i\phi})d\phi = \delta(e^{i\phi}-e^{i\phi_0})d\phi$, on the unit circle. Because of the locality of the measure, Loewner evolution describes the growing curve in the exterior (interior) of the unit disk~\cite{Loe23}. In the case when $e^{i\phi_0}$ is a position of the Brownian particle, the \textit{random} evolution of the curve is known as the Schramm-Loewner evolution (SLE)~\cite{Sch00}. The locality and continuity of the Brownian motion implies the absence of branching of SLE curves.

Another natural examples rely on absolutely continuous measures. In these cases the domain $D_t$ grows at infinitely many points of the boundary simultaneously. In particular, if the density is $\rho_t(w)=|z'_t(w)|^{-2}$, $w\in\p\mathbb D$, then the normal interface velocity~\eqref{v-rho} is proportional to the harmonic measure of the boundary, $v_n(\zeta)\propto |w'_t(\zeta)|$. Hence, this case corresponds to the so-called Hele-Shaw flow or the deterministic LG problem.

The harmonic measure is important in what follows. Thus, it is instructive to recall its definition and basic properties (see Ref.~\cite{GarnettBook} for detail). Let us consider a two-dimensional Brownian motion $(B_s)_{s\geq0}$ starting from $z\in D$, and let $T=T(D)$ be the first exit time. Then, $B_s$ converges to $\hat B_s\in \p D_t$ as $s\uparrow T$, with limit $\hat B_T$ distributed at $\p D$. Then, the harmonic measure, $\omega_{D}(z,V)$, in $D$ from $z$ of a given portion of the boundary $V\in\p D$ is $\omega_{D}(z,V)=  \mathbb P^z[\hat B_T\in V]$.

The harmonic measure is conformally invariant, i.e., $\omega_D(z,V)=\omega_{D'}(f(z),f(V))$, where $f:D\to D'$. When $\p D$ is a Jordan curve, the harmonic measure has a density called a Poisson kernel, $H_{D}(z,s)$, namely, $\omega_{D}(z,|ds|)=H_D(z,s)|ds|$. In particular, when $D$ is the exterior of the unit disk, $\mathbb D$, the Poisson kernel is
\begin{equation}
	H_{\mathbb D}(w,e^{i\phi})=\Re\frac{w+e^{i\phi}}{w-e^{i\phi}},\qquad |w|>1.
\end{equation}
The conformal invariance of the harmonic measure gives the following transformation of the Poisson kernel
\begin{equation}\label{H-transform}
	H_D(z,s)=|f'(s)|H_{D'}(f(z),f(s)),
\end{equation}
where $f:D\to D'$. Thus, if $w_t:D_t\to\mathbb D$ is the conformal map from $D_t$ to the exterior of the unit disk $\mathbb D$, one can obtain an explicit expression for $H_{D_t}$ (denoted by $H_t$ below), namely,
\begin{equation}\label{H-w}
	H_{t}(z,s)=|w'_t(s)|\Re\frac{e^{i\phi}+\xi}{e^{i\phi}-\xi}, \qquad e^{i\phi}=w_t(s),\ \xi=\bar w_t(\bar z)^{-1},
\end{equation}
where bar denotes complex conjugation $\overline{f(z)}=\bar f(\bar z)$.

Below, we consider mainly the inverse conformal map, $w_t:D_t\to\mathbb D$, from the domain $D_t$ to the exterior of the unit disk, $\mathbb D$, in the $w$ plane normalized so that $w_t(\infty)=\infty$ and $w'_t(\infty)=r_t^{-1}$. The Loewner-Kufarev equation for the inverse map can be obtained by using the characteristic equation for~\eqref{z_eq}:
\begin{equation}\label{w_eq}
	\frac{d w_t(z)}{d t}=p_t(w_t(z))w_t(z).
\end{equation}
This equation describes the evolution of the contour in the $w$ plane. Because $\Re p_t(w)<0$, the contour $w_{t+d t}(z)$, where $z\in\p D_t$, lies inside the unit disk, i.e., the contour in the $w$ plane contracts to the origin with time.

It is instructive to recall basics facts about deterministic Hele-Shaw problem first. Afterwards, we will introduce a stochastic LG model in order to study regularized interface dynamics in the Hele-Shaw cell.

\subsection{Deterministic Laplacian growth}\label{s:LG}
The physical formulation of deterministic LG is deceptively simple. Let $D_t^+$ be a simply connected domain occupied by inviscous fluid or gas between two parallel close plates (Hele-Shaw cell). The inviscous fluid is surrounded by a viscous fluid (oil) occupying the rest of the cell, $D_t=\mathbb C\setminus D_t^+$. The viscous fluid is sucked out the cell by several oil wells in $D_t$. The Navier-Stokes equation for a quasistatic motion of a viscous fluid in the Hele-Shaw cell reduces to the Darcy's law: the fluid velocity, $\textbf{v}=-\nabla P_t$, equals the pressure gradient (in scaled units), where $P_t(z,\bar z)$ is pressure in viscous fluid, and $z=x+i y$, $\bar z=x - i  y$ are complex coordinates.

Because of incompressibility, $\nabla\cdot\mathbf{v}=0$, pressure satisfies Laplace equation $\nabla^2 P_t(z)=0$ in $D_t$, except the points with sinks or sources of viscous fluid which provide growth. If to neglect surface tension, $P_t(z)=const$ at the interface $z\in \p D_t$\footnote{Without loss of generality one can set $P_t(z)=0$ at $z\in \p D_t$.}. In the case, when the only oil sink is located at infinity, pressure is a harmonic function in $D_t$, which diverges logarithmically as $z\to\infty$. Thus, it is proportional to the Green's function of the Dirichlet boundary problem, $P_t(z)=-(Q/2\pi)G_t(z,\infty)$, where $G_t(z,s)$ is the Dirichlet Green's function in $D_t$, and $Q$ is the rate of the source~\footnote{The rate of the source is determined by the area of fluid sucked out the cell during a time unit.}.

Recall, that the Green's function $G_t(z,s)$ in the domain $D_t$ is  uniquely determined by the following properties: (a) The function $G_t(z,s)-\log |z-s|$ is symmetric, bounded, and harmonic everywhere in $D_t$ in both arguments; (b) $G_t(z,s)=0$ if $s\in \p D_t$. The Green's function can be written explicitly in terms of the uniformization map of the domain $D_t$. In particular, by using the map $w_t: D_t\to \mathbb D$, one obtains the Green's function in $D_t$:
\begin{equation}\label{G-def}
	G_t(z,s)=\log \left|\frac{w_t(z)-w_t(s)}{1-\overline{w_t(z)}w_t(s)}\right|.
\end{equation}

Now, returning to the Hele-Shaw problem, and taking account of the kinematic identity, which equates the normal interface velocity, $v_n(s)$, and the fluid normal velocity at $s\in\p D_t $, one obtains $v_n(s)=-\p_n P_t(s)$, where $\p_n$ stands for the normal derivative with the unit normal vector $n$ pointing inside the domain $D_t$. Since $P_t(z)\propto G_t(z,\infty)$, from eq.~\eqref{G-def} one determines the normal velocity, $v_n(s)=(Q/2\pi)\log |w'_t(s)|$, of the interface. Hence, the idealized (without surface tension) deterministic LG problem with the oil sink at infinity with the rate $Q$ is similar to the Loewner chain driven by the density~\cite{SB84}:
\begin{equation}\label{rho-lg} 
	\rho_t(e^{i\phi})=(Q/2\pi)|z'_t(e^{i\phi})|^{-2}.
\end{equation}
One can show, that Loewner-Kufarev equation with density~\eqref{rho-lg} is equivalent to the following nonlinear partial differential equation:
\begin{equation}\label{lg-eq}
	\Im\left[\p_t \bar z_t(e^{-i\phi})\p_\phi z_t(e^{i\phi})\right] = Q/ 2\pi.
\end{equation}
This is a classical LG equation, which was intensely studied earlier.

Although the interface dynamics in the deterministic LG problem obeys nonlinear dissipative equation of motion~\eqref{lg-eq}, it is integrable. Namely, the problem possesses infinitely many conservation laws~\cite{Ric72}, and exact solutions can be obtained in a closed form~\cite{SB84}. However, the idealized deterministic LG is ill-posed: the interface develops cusp-like singularities in a finite time~\cite{SB84}. Therefore, the problem needs to be regularized. One possibility is to use a conventional hydrodynamical regularization realized through a surface tension. In this case, the pressure at the interface is proportional to its curvature, $P(s)=-\chi\kappa(s)$, $s\in\p D_t$, where the constant $\chi$ is the coefficient of the surface tension, and the real-valued function $\kappa(s)$ is the curvature of the boundary.\footnote{The curvature of the interface at $s\in\p D_t$ can be written in terms of the uniformization map $w_t$ as follows: $\kappa(s)=\p_n\log \left(|w_t(s)|/|w'_t(s)|\right)$. We also note, that the Hele-Shaw flow regularized by surface tension is equivalent to the Loewner chain with density
\begin{equation}\nonumber
	\rho_t(u)=\frac{1}{|z_t'(u)^2|}\left[\frac{Q}{2\pi}+\chi\Re\left(u\p_u\int_0^{2\pi}\frac{d\phi}{2\pi}\frac{u+e^{i\phi}}{u-e^{i\phi}}\kappa[z_t(e^{i\phi})]\right)\right].
\end{equation}
}. However, nonzero surface tension destroys a rich mathematical structure of idealized LG, and complicates the analytical analysis. This motivates us to consider another possible regularizations of problem.

\subsection{Discrete stochastic Laplacian growth}\label{ss:dslg}

In this section we briefly recall a model for the regularized interface dynamics in the Hele-Shaw cell. The model is based on the stochastic realization of the Darcy's law, namely, diffusion-limited aggregation (DLA)~\cite{WS81}.

Discrete stochastic LG generalizes DLA in what follows. Instead of one particle per time unit the distant source emits $K\geq1$ uncorrelated Brownian particles one-by-one. However, no motion of the boundary occurs until all of them hit the interface. The advance of the interface, $\delta h(s)$, along the normal vector $n(s)$ at the boundary point $s\in\p D_t$ per time unit $\delta t$ is determined by the number of particles, $k(s)$, which hit the given arc of the boundary, $|ds|$, times the linear size, $\hbar^{1/2}$, of the particle:
\begin{equation}\label{deltah}
	\delta h(s) = k(s) \hbar^{1/2}, \quad s\in \p D_t.
\end{equation}
The total number of newly aggregated particles equals $\oint_{\p D}k(s)|d s|=K$.

Different outcomes of the aggregation process determine possible growth scenarios. When $K$ is large, all points of the interface advance simultaneously. Let $v_n(s)=\delta h(s)/\delta t$ be the instantaneous normal velocity of the interface at $s\in \p D_t$. Then, $v_n$ is a random variable with a multinomial distribution~\cite{AMW16}. In the large $K$ limit one can recast the multinomial distribution into Dyson's circular distribution of eigenvalues of random $N\times N$ matrices~\footnote{It is supposed, that the boundary, $\p D_t$, is divided into $N\gg1$ little arcs of the size $\sqrt\hbar$.}, whose eigenvalues are proportional to instantaneous  normal velocities of $N$ infinitesimal arcs of the boundary~\cite{QLG1}. Namely, various scenarios of the discrete stochastic LG obey the Gibbs-Boltzmann statistics:
\begin{equation}\label{P-v}
	\mathbb P[v_n]\propto\exp\left(\sum_{i=1}^N \sum_{j=1}^N v_n(s_i)\log |w_t(s_i)-w_t(s_j)|v_n(s_j)\right).
\end{equation}
The variation of eq.~\eqref{P-v} shows that $\mathbb P$ is maximal when the $v_n=\hat{v}_n$, where $\hat{v}_n(s)\propto|w_t'(s)|$, $s\in\p D_t$. As mentioned previously (see Section~\ref{s:LG}), this process is similar to idealized deterministic LG driven by a single source at infinity.

Fluctuations of $v_n$ around $\hat{v}_n$ are forbidden in idealized deterministic LG, because the viscous fluid is incompressible. However, the discreteness of the aggregation process, and the finite size of the particles effectively lead to the two-fluid model compressible on the microscale of the order $\sqrt\hbar$ in the vicinity of the interface, while in the bulk $\nabla\cdot \textbf{v}=0$. These fluctuations are a source of noise for the LG problem. In Ref.~\cite{Ale19a} we considered possible time evolutions of the fluctuations and showed, that the formation of viscous fingers is an intrinsic feature of the growth process. 

\subsection{Continuous stochastic Laplacian growth}\label{ss:cslg}

\begin{figure}[t]
\centering
\includegraphics[width=.7\columnwidth]{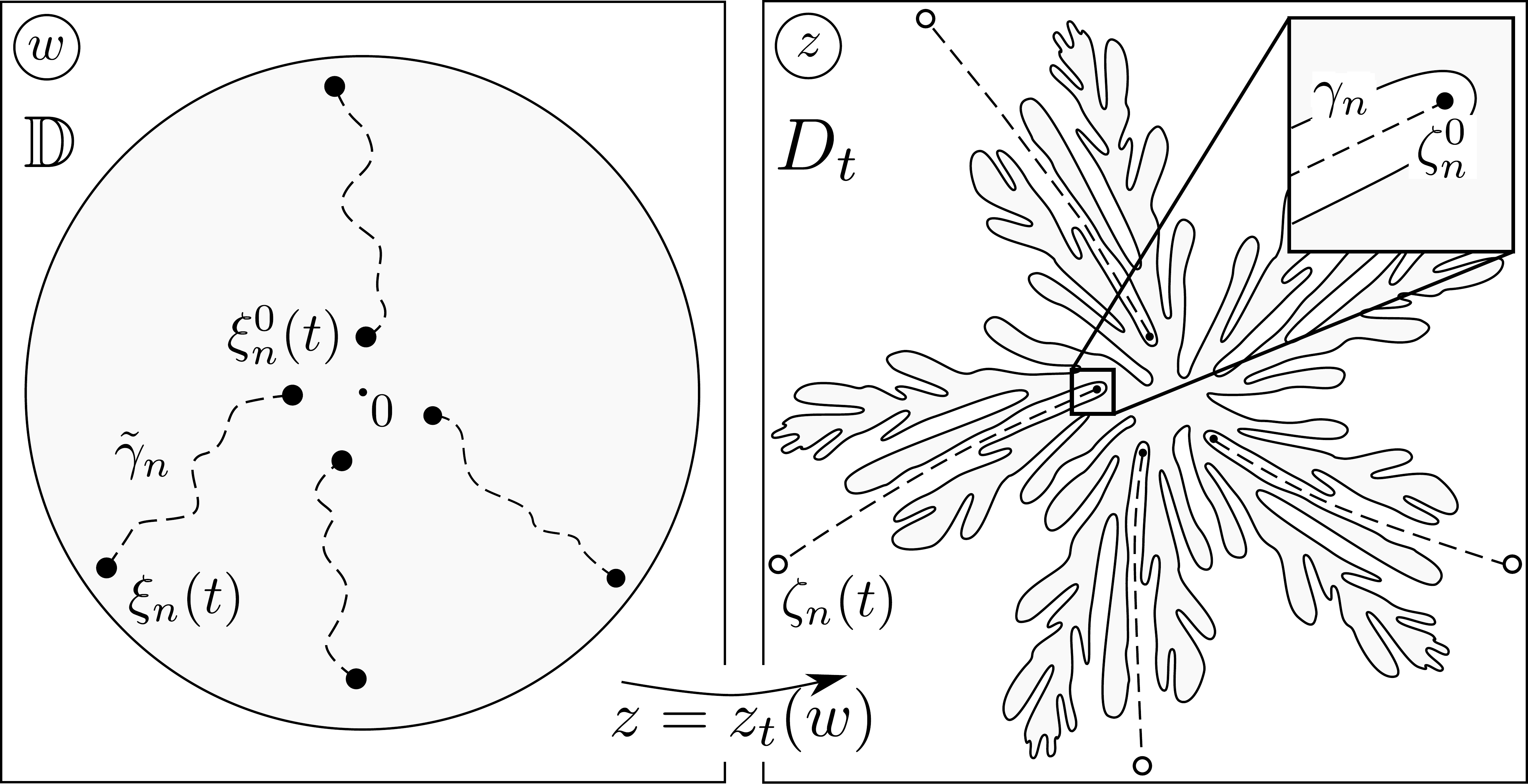}
\caption{\label{cft_domain}
The Loewner chain $(z_t)_{t\geq0}$ describes a family of conformal maps, $z_t: \mathbb D\to D_t$, from the complement of the unit disk $\mathbb D$ in the auxiliary $w$ plane to the domain $D_t$ in the $z$ plane. The map is normalized so that $z_t(\infty)=0$, and $\p_w z_t(\infty)=r_t>0$. The dashed lines on both planes represent the corresponding geometrical objects in both planes: the centerlines of fjords in the $z$ plane, and the branchcuts of $z_t(w)$ inside the unit disk in the $w$ plane. The centerline $\gamma_n$ with the endpoints $\zeta_n^0$ and $\zeta_n(t)$ is a branchcut of the Schwarz function. In stochastic LG the branchcuts $\gamma_n(t)$ evolve with time, because the endpoints $\zeta_n(t)$ moves toward infinity.
}
\end{figure}

We propose the following definition of the stochastic LG model in the continuum limit (see Refs.~\cite{Ale19a,Ale19b} for details). It can be represented as the Loewner chain $(z_t)_{t\geq0}$ generated by the density
\begin{equation}\label{rho-slg}
	\rho_t(e^{i\phi})=\frac{\nu}{|z_t'(e^{i\phi})|^2}\left[\sigma-\frac12\sum_{n=1}^N\alpha_n\Re\frac{e^{i\phi}+\xi_n(t)}{e^{i\phi}-\xi_n(t)}\right],
\end{equation}
where $\sigma=(Q/2\pi\nu) + \sum \alpha_n/2$, so that the total growth rate is $Q$. The quanta of the growth rate, $\nu=\hbar/\delta t$, is the cutoff parameter of the model. The driving processes $\xi_n(t)$ are semimartingales that satisfy the set of coupled stochastic differential equations:
\begin{equation}\label{dxi}
	\frac{d\xi_n}{\xi_n} = -\sigma dq_n + g_n(\bm{\xi},\bar{\bm{\xi}})dq_n + idW,\qquad n=1,2,\dotsc,N,
\end{equation}
with the initial conditions
\begin{equation}\label{xi-def}
	\xi_n(t=0)=\xi_n^0,\qquad n=1,2,\dotsc,N.
\end{equation}

In eq.~\eqref{dxi} we introduced a set of driving processes, $\bm{\xi}=\{\xi_1,\xi_2,\dotsc,\xi_N\}\in\mathbb D^+$, and the random function, $W(q)=\sqrt{\kappa/2}B(q)$, proportional to the standard Brownian motion $B(q)$ with the mean zero, $\mathbb E[W(q)]=0$ (where $\mathbb E$ is the expectation value), and covariance
\begin{equation}
	\Cov[W(q),W(q')]=(\kappa/2)\min(q,q').
\end{equation}

It is assumed that the auxiliary times $q_n(t)$ in the right hand side of eq.~\eqref{dxi} are connected to time $t$ by the following differential equation
\begin{equation}\label{qk}
	dq_n(t)=\nu|w_t'(\zeta_n^0)|^{2}dt.
\end{equation}
The points $\zeta_n^0=z_0(1/\bar\xi_n^0)$, $n=1,2,\dotsc,N$, are the constants of motion representing the positions of the tips of $N$ fjords in the long time asymptotic (see Ref.~\cite{Ale19b} for details).

\begin{remark}\label{r:scaling} Because of stochastic dynamics of the driving processes~\eqref{dxi}, the bottoms of the fjords have universal fractals shapes (see Fig.~\ref{single1}). Hence, the ordinary derivative, $w_t'(s)$, is not well-defined, while the harmonic measure does. The curve $\p D_t$ can be covered by the disks $B(s_i,\varepsilon)$ of radius $\varepsilon$ centered at the points $s_i\in \p D_t$. Let $p(s_i,r)=\omega(\p D_t\cap B(s_i,r))$ be the harmonic measure of the portion of the curve covered by $B(s_i,r)$. One can define a subset $\p D_t^{\alpha}\subset\p D_t$, consisting of the points $s\in\p D_t^{\alpha}$, such that $p(s,r)\sim(r/L)^\alpha$, as $r/L\to0$, where $L$ is the diameter of $D_t$. Hence $|w'_t(\zeta_n^0)|$ in eq.~\eqref{qk} should be replaced by $\varepsilon_n^\alpha$ in the vicinity of fractal boundary, where $\varepsilon_n$ is a distance between $\zeta_n^0$ and the interface. Then, $dq_n(t) \approx \nu \varepsilon_n^{2\alpha} dt$ in the long time asymptotic.
\end{remark}

Finally, the function $g_n(\bm{\xi},\bar{\bm{\xi}})$ in the right hand side of eq.~\eqref{dxi} reads
\begin{equation}\label{g-def}
g_n(\bm{\xi},\bar{\bm{\xi}})=-\frac{\kappa}{2}\xi_n\frac{\p}{\p \xi_n}\log Z_N(\bm{\xi},\bar{\bm{\xi}})+\frac12\sum_{m=1,m\neq n}^N\frac{\xi_n+\xi_m}{\xi_n-\xi_m}+\frac12\sum_{m=1}^N\frac{\xi_n+1/\bar\xi_m}{\xi_n-1/\bar\xi_m},
\end{equation}
where by $Z_N(\bm{\xi},\bar{\bm{\xi}})$ we denoted the following product
\begin{equation}\label{Z-def}
	Z_N(\bm{\xi},\bar{\bm{\xi}}) = \prod_{k<n}^N|\xi_k-\xi_n|^{4/\kappa}\prod_{k\leq n}^N|1-\xi_k\bar\xi_n|^{4/\kappa}.
\end{equation}

\begin{remark}
The function $Z_N$ is similar to the $N$-point correlation function of the vertex operators $V_{\alpha_n,\bar \alpha_n}(\xi_n,\bar\xi_n)$ in the CFT framework~\cite{DiFrancescoBook,Dotsenko84}. By substituting $Z_N$ to eq.~\eqref{g-def} one obtains $g_n=0$, so that the random processes~\eqref{dxi} do not interact with each other. Note, however, that the anzats for the functions $g_n(\bm{\xi},\bar{\bm{\xi}})$ and $Z_N(\bm{\xi},\bar{\bm{\xi}})$ was originally proposed for the correlation functions of primary operators~\cite{Ale19b}. They differ form the correlation function of vertex operators by insertions of the so-called screening operators. However, these insertions do not change the conformal properties of the correlation functions. Hence, the vertex operators can be used to study critical behavior of the correlation functions, relevant to multifractal analysis of the boundary. Another important advantage of~\eqref{Z-def} over the full correlation functions of primary operators is that in the present case the martingales of stochastic LG can be written in terms of elementary functions (see Proposition~\ref{p:1m} below).
\end{remark}

The initial condition for the Loewner-Kufarev equation~\eqref{z_eq} is $z_{0}(w)=w$, i.e., the domain $D^+_{t=0}$ in the unit disk. In deterministic LG~\eqref{lg-eq}, the initial disk continues to stay as a disk with a growing radius.  Complex irregular shapes typically observed at $t > 0$ are explained by tiny uncontrollable initial deviations from a perfect disk, growing because of instability of the process. In stochastic LG the complex shapes are attributed to random fluctuations of pressure, parametrized by the processes $\xi_n(t)$ in eq.~\eqref{rho-slg}.

The ansatz for the driving processes~\eqref{dxi} can by justified by statistical mechanics arguments~\cite{Ale19b}. Namely, by coupling a critical statistical system to $D_t$, and arguing that certain quotients of correlation functions of the corresponding CFT should be martingales with respect to stochastic Loewner chain with density~\eqref{rho-slg}, one can obtain a set of stochastic differential equations for the processes $\xi_n$, $n=1,2\dotsc,N$, consistent with the conformal symmetry of the statistical systems. From this point of view, the function $Z_N$ can be considered as the multi-point correlation function of $N$ degenerate (at the second level) fields in the Coulomb gas representation. Although, we will not follow this point of view here, the CFT interpretation provides a clue to the martingales of stochastic LG.

\subsection{Local fluctuations of  pressure in stochastic Laplacian growth}\label{ss:pressure}

Formally, the Loewner chain with density~\eqref{rho-slg} is similar to the idealized deterministic LG with the oil sink with the rate $Q$ at infinity, and $N$ oil sources with rates $\nu \pi\alpha_n$ at the points
\begin{equation}\label{zeta-def}
	\zeta_n(t)=z_t(1/\bar \xi_n(t)),\quad n=1,2,\dotsc,N.
\end{equation}
If $\zeta_n=const$ ($n=1,2,\dotsc,N$) these points have a clear physical interpretation as the positions of oil wells, which inject the oil into the Hele-Shaw cell with the rates $\nu\pi \alpha_n$. However, stochastic dynamics of the driving processes~\eqref{dxi} results in non-trivial evolution of the oil wells with time. We briefly review the pattern formation in stochastic LG in Appendix~\ref{a:patterns} (see Refs.~\cite{Ale19a,Ale19b} for details). Therefore, one can introduce the (scaled) effective pressure field which drives the interface dynamics:
\begin{equation}\label{P-def}
	P_t(z,\bm{\zeta})=\sigma G_t(z,\infty)-\frac12\sum_{n=1}^N \alpha_n G_t(z,\zeta_n(t)),
\end{equation}
where $G_t$ is the Dirichlet Green's function in $D_t$, defined in Section~\ref{s:LG}, the set of points $\bm{\zeta}=\{\zeta_1,\zeta_2,\dotsc,\zeta_N\}\in D_t$, so that $\xi_n=\bar w_t(\bar \zeta_n)^{-1}$, $n=1,2,\dotsc,N$.

Note, that the only physical oil sink is located at infinity, while the sources at $\zeta_n$'s can be referred to as the \textit{virtual} sources~\cite{AMW16}. The expression for pressure~\eqref{P-def} is valid in the vicinity of the boundary of the domain, because the short-distance regularization suggested by the aggregation model  effectively leads to \textit{compressible} fluid on a small scale of the order of $\hbar$ near the interface due to a finite size of the particles. In the bulk the fluid is incompressible, and the pressure is $P_t(z)=(Q/2\pi)G_t(z,\infty)$. Hence, the virtual sources parametrize local fluctuations of pressure in the vicinity of the boundary of the domain.

The time evolution of the domain, $D_t$, generated by the Loewner equation~\eqref{w_eq}, results in the evolution of the Green's function $G_t$ in $D_t$ with time. It can be determined by using Hadamard's formula, which gives the variation of the Green's function under smooth deformations of the boundary of the domain. Let $d  h(s) = v_n(s) d t$ be the thickness between the curve $\p D_t$ and the deformed curve $\p D_{t+d t}$ counted along the unit normal vector $n(s)$ pointing inside the domain $D_t$ at the point $s\in\p D_t$. Then, the variation of the Dirichlet Green's function reads:
\begin{equation}\label{Hvf}
	d G_t(z,z') = \frac{dt}{2\pi}\oint_{\p D_t}|ds|\p_n G_t(z,s)\p_n G_t(z',s)v_n(s),
\end{equation}
where $\p_n G_t(z,s)$ denotes the normal derivative of the Green's function on the boundary with respect to the second variable, and the unit normal vector points inside the domain $D_t$.

Now, after recalling Hadamard's variational formula for the Dirichlet Green's function the following Lemma can be proved:
\begin{lem}\label{l:dp}
Let us consider the effective pressure field, $P_t(s,\bm{\zeta})$, at $s\in\p D_t$. Then, the variation of $P_t(s,\bm{\zeta})$ with respect to the infinitesimal deformation of the domain generated by the Loewner chain with density~\eqref{rho-slg} reads
\begin{equation}\label{dP}
	d P_t(s,\bm{\zeta})=\left(\sigma-\frac12\sum_{n=1}^N \alpha_n\Re\frac{e^{i\phi}+\xi_n}{e^{i\phi}-\xi_n}\right)\left(\sigma-\frac12\sum_{m=1}^N \alpha_m\Re \frac{e^{i\phi}+\xi_m}{e^{i\phi}-\xi_m}\right)dq(t),
\end{equation}
where $e^{i\phi} = w_t(s)$, $\xi_n=\bar w_t(\bar \zeta_n)^{-1}$, $n=1,2,\dotsc,N$, and $dq(t)=\nu|w_t'(s)|^{2}dt$.
\end{lem}
\begin{proof}
Let us consider the variation of the Green's function with respect to stochastic Loewner flow with density~\eqref{rho-slg}. By taking into account eq.~\eqref{G-def}, one obtains the normal derivative, $\p_n G=2\Re(nG')$, of the Green's function:
\begin{equation}\label{G-H}
	\p_n G_t(\zeta,s)=
	-H_t(\zeta,s)
	,
\end{equation}
where $H_t$ is the Poisson kernel in $D_t$, which can be written explicitly in terms of the uniformization map~\eqref{H-w}. Because of eq.~\eqref{v-rho}, Hadamard's variational formula~\eqref{Hvf} for the Green's function in $D_t$ takes the form:
\begin{equation}\label{dG}
	d G_t(\zeta,s)=\Re\frac{e^{i\phi}+\xi}{e^{i\phi}-\xi}\rho_t(e^{i\phi})dt,
\end{equation}
Hence, the variation of pressure~\eqref{P-def} near the interface $\p D_t$ with respect to stochastic LG with density~\eqref{rho-slg} is given by eq.~\eqref{dP}.

\end{proof}

\section{Martingales of stochastic Laplacian growth}\label{s:martingales}

\subsection{Preliminaries} In this section we propose a family of local martingales closely connected to the stochastic LG problem. Roughly speaking, the martingale, $M_t$, is a random process, for which the conditional expectation value at the next time instant is equal to the present value, $\mathbb E[M_{t_{n+1}}|M_{t_n},\dotsc, M_{t_1}]=M_{t_n}$. Equivalently, one can show that martingales satisfy stochastic differential equations without drift terms, $d M_t=f_t(M_t)dB_t$.

The martingales turns out to be essential objects to study geometrical properties of various two-dimensional fractal sets generated by Loewner chains. Although some martingales for SLE can be obtained explicitly by means of stochastic calculus, the methods of conformal field theory provide a conventional framework to study wide classes of martingales. The connection between SLEs and CFTs relies on the interpretation of SLE curves as the level lines of Gaussian free field, which, in turn, emerges as a continuum description of certain discrete statistical systems. These two-dimensional systems at critical points\footnote{At the critical point the correlation length of the system diverges.} are described by conformal field theories.

Although we will not use CFT methods below, the conformal field theory provides a clue to martingales of stochastic LG. Indeed, let us consider a boundary CFT in $D_t$, which is characterized by the set of scaling operators (primary fields) $\Phi_{h,\bar h}(z,\bar z)$ constituting representations of the Virasoro algebra. The set of conformal dimensions, $(h,\bar h)$, specifies the transformation of the correlation functions with respect to conformal mappings, e.g.,
\begin{equation}
	\langle\prod_{i}\Phi_{h_i,\bar h_i}(z_i,\bar z_i)\rangle_{D_t}=\prod_i (w'_t(z_i))^{h_i}(\bar w'_t(\bar z_i))^{\bar h_i}\langle\prod_{i}\Phi_{h_i,\bar h_i}(w_t(z_i),\bar w_t(\bar z_i))\rangle_{\mathbb D},
\end{equation}
where $w_t:D_t\to\mathbb D$. In the chiral CFT on $\mathbb C$ the transformation law of correlation functions completely factorizes into the variation with respect to the holomorphic coordinates $z_i$, and the variation with respect to the antiholomorphic coordinates $\bar z_i$. Thus, one can consider the coordinates $z$ and $\bar z$ as formally independent.

The holomorphic and antiholomorphic sectors of CFT are no longer independent on the manifolds with boundaries, because the permissible conformal transformations of both sectors must obey certain boundary conditions. For the most conventional domains (either the unit disk or the upper half-plane) the non-chiral $N$-point correlation functions can be obtained by considering chiral $2N$-point function on $\mathbb C$ restricted by certain conditions along the unit circle of the real line. The procedure of extending the domain of the definition of the boundary CFT naturally leads to the Schottky double construction~\cite{SS54}. A possible application of the Schottky double for studying deterministic LG problem can be found in Refs.~\cite{Gustafsson83,Krichever04}. The procedure of doubling of $D_t$ together with the analysis of Ref.~\cite{Ale19b} implies that the martingales of stochastic LG should be the functions of random process on the Schottky double of $D_t$.  Below, we briefly recall some basic notions of the Schottky double construction.

The Schottky double $\hat\Sigma$ of a planar domain with boundary is a compact Riemann surface without boundary endowed with an antiholomorpic involution $j$ on $\hat\Sigma$. In order to make a Schottky double of the bounded domain $D_t$, we take a copy $D_t^*$ of $D_t$, attach $D_t^*$ to $D_t$ along the boundary $\p D_t$, and add the points at both infinities ($\infty,\infty^*$), so that a compact Riemann surface $\hat\Sigma$ is obtained:
\begin{equation}
	\hat\Sigma= D_t\cup\p D_t\cup D_t^*,
\end{equation}
Note, that the Schottky double $\hat\Sigma$ of the disk $D_t$ is a Riemann sphere. If $z\in D_t$ is a holomorphic coordinate on $D_t$, let $z^*\in D^*_t$ denotes the corresponding point on $D^*_t$. The involution $j$ on $\Sigma$ is defined by $j(z)=z^*$, $j(z^*)=z$, and $j(z)=z$ for $z\in \p D_t$. The conformal structure on $D_t$ can be extended across $\p D_t$ to a conformal structure on all of $\Sigma$. The conformal structure on $D^*_t$ is the opposite to that on $D_t$. Hence, the function $z^*\to\bar z$ can be considered as a local coordinate on $D_t^*$.

\begin{figure}[t]
\centering
\includegraphics[width=.35\columnwidth]{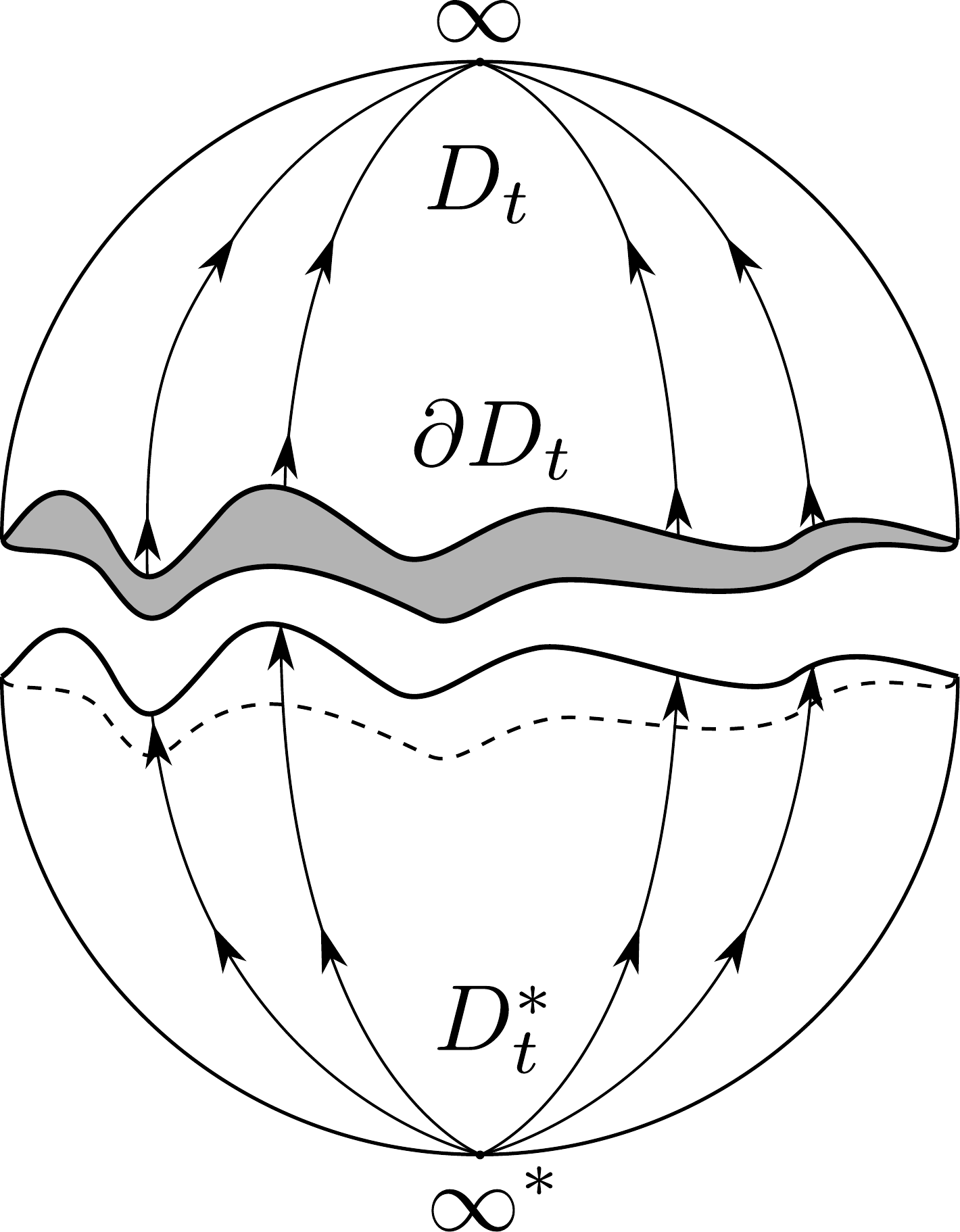}
\caption{\label{double}
The two sides $D_t$ and $D_t^*$ of the Schottky double $\hat\Sigma$ (with the points $\infty$ and $\infty^*$ added) are glued together along the boundary $\p D_t$. The points at the boundary are invariant with respect to the antiholomorphic involution. By the arrows with lines we show schematically the vector field of conformal transformations on the Schottky double, generated by the Loewner flow on $D_t$ and its continuation on $D^*_t$.
}
\end{figure}

\begin{remark}\label{r:barw}
Below, we use the notations $z^*$ and $\bar z$, where $\bar z$ always refers to the complex
conjugate of $z$, while $z^*$ is treated as an independent coordinate of the antiholomorphic side of the Schotky double, which will eventually be set to $\bar z$.
\end{remark}

In order to study martingales of stochastic LG it is convenient to consider the inverse conformal map $w_t: D_t\to\mathbb D$, which satisfies eq.~\eqref{w_eq}. Correspondingly, we will consider the Schottky double $\Sigma=\mathbb D\cup \p \mathbb D\cup\mathbb D^*$, i.e., the Riemann surface obtained by gluing the exterior of the disk $\mathbb D$ and its copy $\mathbb D^*$ along the boundary $\p\mathbb D$. The Loewner-Kufarev equation~\eqref{w_eq} for the chain $(w_t)_{t\geq0}$ generates the infinitesimal conformal transformation $w_t\to w_{t+dt}=w_t+\epsilon(w_t)$ of $\mathbb D$, where $\epsilon(w_t)=w_t p_t(w_t)dt$. Since $\Re p(w_t)<0$ at $w_t\in\p \mathbb D$, the unit circle contracts toward the origin with respect to the infinitesimal deformation $\epsilon(w_t)$. Similarly, one can say that the time evolution of the contour on the plane can be considered as a certain dynamics of the curve $\p\mathbb D$ on the double $\Sigma$.

Let us continue the Loewner flow $\epsilon(w_t)$ from $\mathbb D$ to the antiholomorphic side, $\mathbb D^*$, of the double $\Sigma$ by requiring the flow to be smooth upon crossing $\p \mathbb D$. By taking into account that the function $p_t(w)$ has a jump across the unit circle~\eqref{p-limits}, and noting that $\epsilon(w)/w=\bar w \epsilon(w)$ when $w\bar w=1$ and $p_t(1/\bar w)=-\bar p_t (\bar w)$, one concludes that the Loewner flow on $\mathbb D^*$ in the vicinity of $\p\mathbb D$ is $\varepsilon(w^*)=-\bar \epsilon(w^*)$. Thus, the Loewner flows on both side of the Schottky double $\Sigma$ read:
\begin{equation}\label{barw_dt}
	d w_t(z)=w_t(z) p_t(w_t(z))dt,\qquad d w^*_t(z^*)=- w^*_t(z^*)p_t^*(w^*_t(z^*))dt,
\end{equation}
where $p_t(w)$ is defined by eq.~\eqref{p-def}, and the function $p_t^*(w^*)$ can be written in terms of the following contour integral:
\begin{equation}
	p_t^*(w^*)= \int_0^{2\pi}\frac{d\phi}{2\pi}\frac{e^{- i\phi}+w^*}{e^{-i\phi}-w^*}\rho_t(e^{i \phi}).
\end{equation}
The Loewner density for stochastic LG is given by eq.~\eqref{rho-slg}, where the real part of the analytic function can be replaced by $2\Re f(w) = f(w)+f^*(w^*)$. Indeed, at the time instant $t$, when the two sides of the Schottky double are glued along the unit circle, the involution between the sides identifies $\xi^*$ with the complex conjugated $\bar\xi$. 

Besides, in addition to the driving processes~\eqref{dxi} (where the complex conjugated coordinates, $\bar \xi_n$, should be replaced by the coordinates, $\xi^*_n$, on the antiholomorphic side $\mathbb D^*$ of the double), the stochastic random processes $\xi^*_n\in\mathbb D^*$, $n=1,2,\dotsc,N$, should be specified. In order to be compatible with the flow~\eqref{barw_dt} on $\Sigma$ and eq.~\eqref{dxi}, the processes $\xi_n^*$ are expected to be semimartingales, which satisfy the following set of coupled differential equations:
\begin{equation}\label{dbarxi}
		 \frac{d\xi_n^*}{\xi_n^*}=  \sigma dq_n - g_n(\bm{\xi},\bm{\xi}^*)dq_n - i d W,\qquad n=1,2,\dotsc,N,
\end{equation}
where the function $g_n(\bm{\xi},\bm{\xi}^*)$ is given by eq.~\eqref{g-def}, provided that the complex conjugated coordinates $\{\bar\xi_n\}_{n=1}^N$ are replaced by $\{\xi^*_n\}_{n=1}^N$ on the antiholomorphic side of the double $\mathbb D^*$.

Below, we will propose a family of functions on $\Sigma$, which are the martingales with respect to the stochastic Loewner flow~\eqref{barw_dt} on the Schottky double $\Sigma$, driven by the random processes~\eqref{dxi},~\eqref{dbarxi}.

\subsection{Exponential martingales on the Schottky double.}
In this section we determine a family of martingales on the Schottky double with respect to the Loewner flows~\eqref{barw_dt}. Let us proof auxiliary results first. In Appendix~\ref{a:patterns} we recall the basic features of Loewner chains generated in stochastic LG (see Refs.~\cite{Ale19a,Ale19b} for details). In particular, on can argue, that in the discrete time approximation the Loewner chain $(z_t)_{t\geq0}$ can be written as a finite sum of logarithmic terms (see eq.~\eqref{a:zI} below). Then, we can prove the following Lemma.

\begin{lem}\label{l:w'}
Let $\delta_{n}(t)=\dist(\xi_{n}^0(t),\p \mathbb D)$ be the the distances between the poles, $\xi_n^0(t)$, of $z'_t(w)$ and $\p \mathbb D$ at time instant $t$. Let us define the following set of points
\begin{equation}\label{S-def}
	S=\{z\in D_t|\forall z:\exists \delta_{n}(t)\ dist(w_t(z),\p\mathbb D)<\delta_{n}(t),\ \p_\tau|w_t'(z)|=0\},
\end{equation}
where $\p_\tau f =\Im[n(z)\p_z f]$, $n(z_t(w))=w z'_t(w)/|z'_t(w)|$. Suppose that $s\in S$ and consider the logarithmic function of conformal factors, $\log (w_t'(s)w_t'^*(s^*))$, where $s^*=j(s)\in \mathbb D^*$. Then, the variation of the logarithmic function with respect to the infinitesimal conformal transformations~\eqref{barw_dt} of the Schottky double reads
\begin{equation}\label{log_ww*}
	d\log \left(w_t' w^{*\prime}_t\right) =d\log (w_t w^*_t)-2\sum_{n=1}^N \alpha_n\left[ \frac{w\xi_n}{(w-\xi_n)^2}-\frac{ w^* \xi^*_n}{(w^*-\xi^*_n)^2}\right]dq.
\end{equation}
The auxiliary time, $q$, is related with the time, $t$, by $dq=\nu|w'_t(s)|^2 dt$.
\end{lem}
\begin{proof}
By differentiating eqs.~\eqref{barw_dt} with respect to $s$ and $s^*$, one obtains the following variation of the logarithms of the conformal factors:
\begin{equation}
	\frac{d\log \left(w_t' w^{*\prime}_t\right)}{dt}=p(w_t) - p^*( w^*_t) + w_tp'(w_t)-w^*_t p^{*\prime}(w_t^*).
\end{equation}

Let us consider the set $S$ defined in~\eqref{S-def}. In the previous works~\cite{Ale19a,Ale19b} and Appendix~\ref{a:patterns} we argued, that the poles, $\xi_n^0$, of $z'_t(w)$ tend to $\p \mathbb D$ following an exponential law in the long time asymptotic, namely, $\dist(\xi_n^0(t),\p \mathbb D)\propto \exp(-r_t/|c_{n,0}|)$, where $r_t$ is the conformal radius, and $c_{n,0}$ is a constant (see eq.~\eqref{a:c-def} below). At the same time, from the results of Ref.~\cite{DMW98} and a similarity of stochastic LG and the deterministic LG with stochastically moving sources it follows, that $\dist(z_t(1/\bar \xi_n^0(t)),\p D_t)= c_{n,0}+O(1/r_t)$, where $z_t(1/\bar \xi_n^0)$ is the tip of the fjord. Hence, for the points $s$ from the line segment $s\in[z_t(e^{i \arg \xi_n^0}),z_t(1/\bar \xi_n^0)]$ we have $\dist(w_t(z),\p \mathbb D)<\dist(\xi_n^0(z),\p \mathbb D)$. Besides, the harmonic measure reaches local minima, $\p_\tau |w'_t(z)|=0$ (where $\p_\tau$ is the tangential derivative along the interface), at the bottoms of the fjords near the points $z=z_t(e^{i \arg \xi_n^0})$. Therefore, the set $S$ is not empty.

Let us consider the point $s\in S$, so that $w_t(s)$ tends to $\p \mathbb D$ following an exponential law as $r_t\to \infty$. In this case, the difference between the integrals, $ w_tp'-w^*_t p^{*\prime}$, reduces to a single contour integral around $w_t(s)$ (which can be considered as the boundary point $w_t(s)\in\p\mathbb D$):
\begin{equation}\label{dlog_w}
	\frac{d\log \left(w_t' w^{*\prime}_t\right)}{dt}=p(w_t) - p^*( w^*_t) + I(w),
\end{equation}
where
\begin{equation}\label{I_comp}
	I(w) = -2 w \oint_w\frac{du}{2\pi i}\frac{\rho_t(u)}{(u-w)^2},
\end{equation}
Note, that the density $\rho_t(u)$ depends both on the holomorphic, and anti-holomorphic coordinates~\eqref{rho-slg}. However, at the unit circle one has $\bar u=1/u$, and, therefore, $\rho_t(u)$ can be analytically continued away from the unit circle unless one encounters a critical point of $z'_t(w)$. Therefore, on can use residue theorem to evaluate the integral $I(w)$.

Since $s\in S$ , and $w_t(s)$ is closer to $\p\mathbb D$ than any critical point of $z'_t(u)\bar z'_t(1/u)$, only the residue at $w$ contributes to the contour integral~\eqref{I_comp}. Let us consider the contributions, $I=I_1+I_2$, which appear upon differentiating the factor $|z_t'(w)|^{-2}$, and the expression in the right hand side of eq.~\eqref{rho-slg} correspondingly,
\begin{subequations}
	\begin{gather}
		\label{I1}
	I_1(w) = 4i w\rho_t(w)\Im\left[w\p_w (z'_t(w)z_t'^*(w^*))\right],\\
	\label{I2}
	I_2(w) = -\frac{\nu}{|z'_t(w)|^2}\sum_{n=1}^N \alpha_n\left[\frac{w\xi_n}{(w-\xi_n)^2}-\frac{ w^* \xi^*_n}{(w^* - \xi^*_n)^2}\right],
	\end{gather}
\end{subequations}
where $\bar z_t'(w^{-1})$ denotes $d\bar z_t(u)/du$ at the point $u=w^{-1}$, and we took into account that $w^{-1}=\bar w$ on the unit circle. The expression in parenthesis in the right hand side of eq.~\eqref{I1} can be recast in the form $-2|w'_t(s)|^{-4}\Im\left(n(s)\p_s|w_t'(s)|\right)$, where $n(z_t(w))=w z'_t(w)/|z'_t(w)|$ is the unit normal vector at $s\in \p D_t$. Therefore, one obtains
\begin{equation}\label{I1-calc}
	I_1(w_t(s))=-8 i w_t(z) \frac{\rho_t(w_t(s))}{|z_t'(w)|^2}\p_\tau|w_t'(z)|,
\end{equation}

Hence, the contribution~\eqref{I1-calc} to the integral $I(w)$ vanishes provided that $s\in S$ is a local extremum of the harmonic measure. In this case the value of the integral $I(w)$ is given by~\eqref{I2} solely. Finally, because of eqs.~\eqref{barw_dt} one can show that the difference $p(w_t)-p^*(w_t^*)$ is equal to the variation of $\log (w_t w_t^*)$.

\end{proof}

In the following Lemma we introduce a certain function of the random processes, $\xi$ and $\xi_n^*$, on the Schottky double $\Sigma$, and study its variation with respect to the Loewner flow.

\begin{lem}\label{l:h}
Let us introduce the auxiliary logarithmic functions,
\begin{equation}
\begin{gathered}
	\mathfrak g^{(1)}(w,w^*,\bm{\xi})=\sum_{n=1}^N \alpha_n\log (w-\xi_n)(1- w^* \xi_n)-\sum_{n=1}^N \alpha_n\log \xi_n,\\
	\mathfrak g^{(2)}(w ,w^*, \bm{\xi}^*)=\sum_{n=1}^N \alpha_n\log (w^* -\xi_n^*)(1-w \xi_n^*)-\sum_{n=1}^N\alpha_n\log \xi^*_n,
\end{gathered}
\end{equation}
of the coordinates $(\xi,w:=w_t(s))\in\mathbb D$ and $(\xi^*,w^*: =w^*_t(s^*))\in{\mathbb D}^*$, where $(w_t)_{t>0}$, $(w^*_t)_{t>0}$ are stochastic Loewner chains~\eqref{barw_dt} generated by the random processes~\eqref{dxi},~\eqref{dbarxi}. By $\mathfrak h$ we denote the following linear combination of the functions $\mathfrak g^{(1)}$, $\mathfrak g^{(2)}$, and the logarithm of conformal factors, $\log \left(w_t'w_t'^*\right)$:
\begin{equation}\label{h-def}
	\mathfrak h(w, w^*;\bm{\xi},\bm{\xi}^*)=\frac{1}{\sqrt\kappa}\mathfrak g^{(1)}(w,w^*,\bm{\xi})-\frac{1}{\sqrt\kappa}\mathfrak g^{(2)}(w,w^*,\bm{\xi}^*)+\frac{\sqrt\kappa}{4}\log \frac{w' w'^*}{w w^*}.
\end{equation}

Suppose that $s\in S$ is a point from the set $S$ introduced in Lemma~\eqref{l:w'}. Then, the variation of $\mathfrak h$ with respect to the infinitesimal deformation generated by stochastic LG with density~\eqref{rho-slg} reads
\begin{multline}\label{dh}
	d\mathfrak h = \frac{1}{\sqrt\kappa}\sum_{n=1}^N \alpha_n\Re\left(\frac{w+\xi_n}{w-\xi_n}\right)\left[\left(-2\sigma+\sum_{m=1}^N\alpha_m \Re\frac{w+\xi_m}{w-\xi_m}\right)dq-\right.\\
	\left.\phantom{\sum_{m=1}^N} + 2 (\sigma - g_n(\bm{\xi},\bar{\bm{\xi}}))dq_n\right] - \frac{2i}{\sqrt\kappa}\sum_{n=1}^N \alpha_n\Re\left(\frac{w+\xi_n}{w-\xi_n}\right) dW,
\end{multline}
where $dq=\nu|w'(s)|^2 dt$, $dq_n=\nu|w'(\zeta_n^0)|^2dt$, and $2\Re f(w)=f(w)+f^*(w^*)$.
\end{lem}
\begin{proof} Let us recall  \^{I}to's lemma for finding the differential  of a time-dependent function of stochastic processes. Suppose that $x_t$ is a random process  that satisfies the stochastic differential equation: $dx_t=\mu_t dt+\sigma_tdB_t$, where $B_t$ is a standard Brownian motion, i.e., $\mathbb E[B_t]=0$, and $\Cov[B_t,B_{t'}] = \min(t,t')$. Then, for any twice-differentiable scalar function $f(t,x_t)$ one has
\begin{equation}\label{Ito}
	df=\left(\p_t f+\mu_t\p_x f+(\sigma_t^2/2)\p_{xx}f\right)dt+\sigma_t(\p_xf)dB_t.
\end{equation}

One can apply \^{I}to's formula~\eqref{Ito} in order to determine the variations of $\mathfrak g^{(1)}$ and $\mathfrak g^{(2)}$ on $\Sigma$ with respect to the infinitesimal conformal transformation generated by Loewner flows~\eqref{barw_dt}:
\begin{equation}\label{dg12}
	\begin{aligned}
	d\mathfrak g^{(1)}=&\sum_{n=1}^N \alpha_n\left[-\frac{w+\xi_n }{w-\xi_n}idW+(\sigma	- g_n(\bm{\xi},\bm{\xi}^*))\frac{w+\xi_n}{w-\xi_n}dq_n\right.+\\
	& \left.+\frac{\kappa}{2}\frac{\xi_n\xi_n}{(w-\xi_n)^2} dq_n + \Re p_t(w)\frac{w+\xi_n}{w-\xi_n}dt + i\Im p_t(w)dt\right],
	\\
	d\mathfrak g^{(2)}=& \sum_{n=1}^N \alpha_n \left[\frac{w^*+\xi^*_n}{w^*-\xi^*_n}i dW-(\sigma-g_n(\bm{\xi},\bm{\xi}^*))\frac{w^*+\xi_n^*}{w^*-\xi_n^*}dq_n + \right.\\
	& \left.+\frac{\kappa}{2}\frac{\xi_n^*\xi_n^*}{(w^*-\xi_n^*)^2}dq_n - \Re p_t(w)\frac{w^*+\xi_n^*}{w^*-\xi_n^*}dt+i  \Im p(w)dt\right],
	\end{aligned}
\end{equation}
Hence, from eqs.~\eqref{dg12} and Lemma~\ref{l:w'} the variation of $\mathfrak h$ follows.
\end{proof}

\begin{remark} Below, it will be assumed that in the long time asymptotic the auxiliary times $q,q_1,q_2,\dotsc,q_N$ can be chosen in such a way that $dq=dq_1=\cdots=d q_N$. Indeed, $dq=\nu|w_t'(s)|^2 dt$, and $s$ can be chosen from the subsets of the interface, $\p D_t^{\alpha}\subset\p D_t$, where the harmonic measure scales with the exponent $\alpha$. Besides, $dq_n=\nu|w_t'(\zeta_n^0)|^2dt$, and one can choose the initial conditions in such a manner, that near the tips of the fjords, $\zeta_n^0$'s, the harmonic measure scales with the exponent $\alpha$ in the long time asymptotic (see Remark~\ref{r:scaling}). Hence, we will drop the subscripts $n$ to $q_n$'s in what follows.
\end{remark}

Note, that the variations~\eqref{dg12} are not symmetric with respect to the interchange $(\xi,w)\to(\xi^*,w^*)$. Moreover, by taking into account that $\Re p_t(w)=-\rho_t(w)$, where $\rho_t(w)$ is given by~\eqref{rho-slg}, one obtains the coefficients, $(\kappa/2-1)$ and $(\kappa/2+1)$, in front of the second order poles, $(w-\xi_n)^{-2}$ and $(w^*-\xi^*_n)^{-2}$, respectively. Thus, these terms do not vanish simultaneously for a specific value of the diffusion coefficient $\kappa$.

Now, we consider the exponential function, $\exp(\gamma \mathfrak h)$, of random processes generated by stochastic Loewner flow~\eqref{barw_dt} on the Schottky double. The following proposition shows, that for a specific value of the coupling parameter, namely, $\gamma=\kappa^{-1/2}$, the random processes are martingales of stochastic Loewner flow on $\Sigma$.

\begin{prop}\label{p:1m}

Let us consider the exponential function of the random process $\mathfrak h$ introduced in Lemma~\ref{l:h}:
\begin{equation}\label{M-def}
	M_t(w;\bm{\xi}):=M_t(w,w^*;\bm{\xi},\bm{\xi}^*)=e^{(1/\sqrt\kappa)\mathfrak h(w,w^*;\bm{\xi},\bm{\xi}^*)},
\end{equation}
where $(w_t)_{t>0}$, $(w^*_t)_{t>0}$ are stochastic Loewner chains~\eqref{barw_dt} generated by the random processes~\eqref{dxi},~\eqref{dbarxi}. The function $M_t(w;\bm{\xi})$ is a local martingale satisfying the following stochastic differential equation:
\begin{equation}\label{dM}
	dM_t(w;\bm{\xi}) = -\frac{2i}{\sqrt\kappa} M_t(w;\bm{\xi})\sum_{n=1}^N \alpha_n\Re\left(\frac{w+\xi_n}{w-\xi_n}\right)dW,
\end{equation}
where $2\Re f(w)=f(w)+f^*(w^*)$.
\end{prop} 
\begin{proof}
The differential of the exponential function, $\exp (f(t,x_t))$, of the random process $x_t$, that satisfies the stochastic differential equation, $dx_t=\mu_tdt+\sigma_tdB_t$, can be determined by using \^{I}to's Lemma~\eqref{Ito}:
\begin{equation}\label{Ito-exp}
	de^{f}=e^{f}\left[(\p_t f+\mu_t\p_x f+(\sigma_t^2/2)(\p_{xx}f+(\p_x f)^2)dt+\sigma_t(\p_xf) dB_t\right].
\end{equation}
By applying this formula to the exponential function~\eqref{M-def} one obtains eq.~\eqref{dM}.
\end{proof}

In the next section we discuss a connection between a family of exponential martingales, $M_t(w,\bm{\xi})$, of stochastic Loewner flow on the Schottky double $\Sigma$, and the stochastic LG of the interface $\p D_t$. The connection relies by the formula, which connects the Hadamard's variational formula for the Green's function to the covariance of the martingales~\eqref{M-def}.

\subsection{One-point exponential martingales and elementary deformations} First, let us consider a family of the so-called elementary deformations of the domain~\cite{KMZ05}. The elementary deformation of the domain  with the base point $\zeta\in D_t$ is such deformation of $D_t$, that the displacement, $\delta_{\zeta} h(s)$, of the infinitesimal portion of the boundary at $s\in\p D_t$ along the unit normal vector pointing inside $D_t$ is proportional to the harmonic measure of that portion of the boundary as viewed from the point $\zeta$:
\begin{equation}\label{delta-h}
	\delta_{\zeta} h_t(s)=-\frac{\nu \delta t}{2}\p_n G_t(\zeta,s),\qquad s\in\p D_t,\quad \zeta\in D_t,
\end{equation}
where $\nu \delta t=\delta A_t$ is the increment of the area, $A_t$, of the domain, and $\p_n G_t(\zeta,s)$ is the normal derivative of the Dirichlet Green's function with respect to the second variable. The elementary deformation of the domain can be considered as the result of the idealized deterministic LG with the oil sink with rate $\nu$ placed at the point $\zeta\in D_t$ instead of infinity.

Let $t_k$ be the harmonic moments of the domain $D_t$ with respect to the basis $z^{-k}/k$:
\begin{equation}
	t_k=-\frac{1}{\pi k}\int_{D_t}z^{-k}d^2z, \qquad k=1,2,\dotsc
\end{equation}
We also introduce $\bar t_k$ to be the complex conjugated moments, and refer to $t_0=(1/\pi)\int_{D^+_t}d^2z$ as to the area (divided by $\pi$) of the bounded domain $D^+_t=\mathbb C\setminus D_t$. These harmonic moments are nothing but the coefficient of the Taylor expansion of the logarithmic gravity potential of the domain, $\Pi(z)=(-2/\pi)\int_{D_t}\log |z-z'|d^2z$, $z\in D^+_t$. 

The harmonic moments $(t_0,t_1,\bar t_1,\dotsc)$ form a set of local coordinates in the space of smooth closed planar curves~\cite{KMZ05,Takhtajan2001}. Elementary deformations of the domain $\delta_\zeta$ allow one to introduce the vector field,
\begin{equation}
	\nabla(\zeta) = \p_{t_0} + \sum_{k\geq1}\left(\frac{\zeta^{-k}}{k}\p_{t_k}+\frac{\bar \zeta^{-k}}{k}\p_{\bar t_k}\right),
\end{equation}
which spans the (complexified) tangent space to the space of curves. The variation of any functional of the domain $F(t_0,t_1,\bar t_1,\dotsc)$ with respect to the elementary deformation with the base point $\zeta$ can be considered as the action of the vector field $\nabla(\zeta)$: $\delta_\zeta F=\nu\nabla(\zeta)F \delta t$.

The elementary deformations plays an important role for studying the integrable structure of deterministic LG. The variation of the Green's function, $\nabla(\zeta)G_t(z,z')$, with respect to the elementary deformation~\eqref{delta-h} can be obtained by using Hadamard's variational formula:
\begin{equation}\label{triple}
	\nabla(\zeta)G(z,z') = -\frac{1}{4\pi}\oint_{\p D_t}\p_n G_t(\zeta,s)\p_n G(z,s)\p_n G(z',s)|ds|.
\end{equation}
Remarkably, this expression is symmetric with respect to the permutations of the points $\zeta$, $z$, and $z'$: $\nabla(\zeta)G(z,z')=\nabla(z)G(z',\zeta)=\nabla(z')G(\zeta,z)$. This exchange relation has a form of integrability condition, and possesses a rich underlying algebraic structures~\cite{Krichever04}. It is a key relation which relates the Dirichlet problem and deterministic LG problem with the universal Whitham hierarchy~\cite{Krichever94}, known as dispersionless Toda lattice (see Refs.~\cite{Krichever04,MWZ02}).

Let us consider the variation of the Green's function, $G(z,s)$, at the boundary of the domain, $z\in D_t$, $s\in \p D_t$ with respect to the elementary deformation with the base point $\zeta\in D_t$. From the Hadamard's formula~\eqref{triple} it follows:
\begin{equation}\label{nabla-G}
	\nabla(\zeta)G_t(z,s)=\frac{1}{2}\p_n G_t(\zeta,s)\p_n G_t(z,s),\qquad s\in \p D_t.
\end{equation}
Thus, the variation of the Green's function at the boundary is given by the product of two Poisson kernels, $H_t(\zeta,s)$ and $H_t(z,s)$, defined by eq.~\eqref{H-w}. Now, the following proposition can be formulated.

\begin{prop}\label{p:covdM}
Let us consider elementary deformations of the domain generated by stochastic Loewner flows with the densities $\rho_t^{(1)}$ and $\rho_t^{(2)}$, where
\begin{equation}
	\rho_t^{(n)}(z)=|z'_t(e^{i\phi})|^{-2}\left(\sigma - \frac12\alpha_n\Re\frac{e^{i\phi}+\xi_n}{e^{i\phi}-\xi_n}\right).
\end{equation}
Let $M_t(w,\xi_1)$, $M_t(w,\xi_2)$ be the one-point martingales of the corresponding flows, where $w=w_t(s)$, and $s\in S$ introduced in Lemma~\ref{l:w'}. Then, the covariance of $dM_t(w,\xi_1)$ and $dM_t(w,\xi_2)$ is proportional to the variation of the Dirichlet Green's function $G_t(\zeta_2,s)$ with respect to the elementary deformation of the domain with the base point $\zeta_1$:
\begin{equation}
	\Cov[d M_t(w,\xi_1),d M_t(w,\xi_2)]=-\alpha_1\alpha_2 M_t(w,\xi_1) M_t(w,\xi_2) \nabla(\zeta_1)G_t(\zeta_2,s)dt.
\end{equation}
where $\zeta_1=z_t(1/\bar \xi_1)$, and $\zeta_2=z_t(1/\bar \xi_2)$.
\end{prop}
\begin{proof}
Because the driving process $W$ is proportional to standard Brownian motion in the auxiliary time $q$ one has $\mathbb E[dWdW]=(\kappa/2)dq$, where $\mathbb E$ denotes the expectation value. Therefore, the expectation value of the product of the one-point martingales reads:
\begin{multline}\label{EdM}
	\mathbb E[d M_t(w,\xi_1) d M_t(w,\xi_2)]= -\frac{\alpha_1 \alpha_2}{2} M_t(w,\xi_1)M_t(w,\xi_2)\times\\
	\times\Re\left(\frac{w+\xi_1}{w-\xi_1}\right)\Re\left(\frac{w+\xi_2}{w-\xi_2}\right)dq,
\end{multline}

Since $\p_n G$ equals the Poisson kernel~\eqref{G-H}, one concludes that the variation of the Green's function with respect to the elementary deformation~\eqref{nabla-G} is proportional to the product of the Poisson kernels, which appear in the right hand side of eq.~\eqref{EdM}. Because $M_t$'s are the martingales, the covariance is given by the expectation value: $\Cov[d M_t, d M_t]=\mathbb E[d M_t d M_t]$.
\end{proof}

\subsection{Multi-point martingales and local fluctuations of pressure}

Thus far, we have considered the sink at infinity with the rate $Q/2\pi$ as the deterministic sink, i.e., as the sink with the fixed position. However, the stochastic LG model, introduced in Section~\ref{ss:cslg}, can be easily generalized to the case of the source, which moves stochastically with time near infinity.

Henceforth, we assume that the interface dynamics is generated by $N+1$ random processes $\bm{\xi}=(\xi_0,\xi,\dotsc,\xi_N)\in\mathbb D^+$, with the rates $(\alpha_0,\alpha_1,\dotsc,\alpha_n)$, and $\alpha_0=-2\sigma$. The stochastic differential equation for the driving processes~\eqref{dxi} should be modified in a trivial way. One replaces the constant drift toward the origin, $-\sigma dq_n$, with the corresponding interaction terms of the processes $\xi_n$ with $\xi_0$. The system of coupled stochastic differential equations takes the form:
\begin{equation}
	d\log \xi_n=-\sigma(1-\delta_{n,0}/2)dq_n + g_n(\bm{\xi},\bar{\bm{\xi}})dq_n + idW,\quad n=0,1,\dotsc,N.
\end{equation}
Here, the function $g_n(\bm{\xi},\bar{\bm{\xi}})$ reads
\begin{equation}
	g_n(\bm{\xi},\bm{\xi})=-\frac{\kappa}{2}\xi_n\frac{\p}{\p \xi_n}\log Z_{N+1}(\bm{\xi},\bar{\bm{\xi}})+\frac12\sum_{m=0,m\neq n}^N\lambda_m \frac{\xi_n+\xi_m}{\xi_n-\xi_m}+\frac12\sum_{m=0}^N \lambda_m\frac{\xi_n+1/\bar\xi_m}{\xi_n-1/\bar\xi_m},
\end{equation}
with $\lambda_0=-\sigma$, $\lambda_n=1$, $n=1,2,\dotsc,N$, and the function $Z_{N+1}$ is given by the product
\begin{equation}
	Z_{N+1}(\bm{\xi},\bm{\xi}) = \prod_{k<n}^N|\xi_k-\xi_n|^{4/\kappa}\prod_{k\leq n}^N|1-\xi_k\bar\xi_n|^{4/\kappa}\prod_{k=1}^{N}|\xi_k|^{-4\sigma/\kappa},
\end{equation}
where we set $\xi_0=0$. 

As explained in the beginning of Section~\ref{ss:pressure}, the dynamics of the interface $\p D_t$ in the $z$ plane is generated by the set of virtual sources $\bm{\zeta}=(\zeta_0,\zeta_1,\dotsc,\zeta_N)\in D_t$, where $\zeta_n=z_t(1/\bar \xi_n)$, $n=0,1,\dotsc,N$. The effective pressure field in the vicinity of the interface is given by the linear combination of the Green's functions~\eqref{P-def}. Therefore, we have the following corollary of Lemma~\ref{p:covdM}.

\begin{cor}
Let us consider the function~\eqref{h-def} of $N+1$ random processes $\xi_n$, $n=0,1,2\cdots,N$. Suppose that $\lambda_0=-\sigma$ and $\mathbb E(\xi_0)=0$. Then, the variance of the exponential martingale, $M_t(w,\bm{\xi})$ determines the variation of pressure $d P_t(z,\bm{\zeta})$ with respect to stochastic LG,
\begin{equation}\label{c:cov}
	d\Cov[M_t(w,\bm{\xi}),M_t(w,\bm{\xi})]=M^2_t(w,\bm{\xi})d P_t(s,\bm{\zeta}),
\end{equation}
where $w=w_t(s)$, $\xi_n=1/\bar w_t(\bar \zeta_n)$, and  $s\in S$ introduced in Lemma~\ref{l:w'}.
\end{cor}

\begin{proof}
The statement immediately follows from Lemma~\ref{l:dp} and Proposition~\ref{p:covdM}.
\end{proof}

\section{Conclusion and discussion}\label{s:conclusion}
In conclusion, we briefly summarize the main result of this work. We studied stochastic LG problem in the framework of Loewner-Kufarev equation drive by the nonlocal random measure on the unit circle~\eqref{rho-slg}. Because of intrinsic instability of the growth process, tiny fluctuations of pressure in the vicinity of the interface results in the formation of universal patterns with deep fjords in the long time asymptotic, which can be studied both numerically and analytically~\cite{Ale19a}. These patterns are closely related with the so-called logarithmic solutions to deterministic LG problem, which are free of finite-time singularities (cusps) on the boundary.

In Ref.~\cite{Ale19b} it was shown, that certain quotients of correlation functions of boundary CFTs do not change with respect to conformal transformations generated by the stochastic Loewner flow. Hence, these functions are martingales of stochastic LG. In this work we studied martingales by using conventional methods of stochastic calculus. In particular, we proposed a family of functions~\eqref{M-def} on the Schottky double of the exterior domain $D_t$, and studied their variation with respect to stochastic Loewner flow. The proposed family of martingales are closely connected to CFT correlation functions in the Coulomb gas formalism. A direct connection between the martingales and stochastic LG problem relies on Hadamard's variational formula for the Green's function of the Dirichlet boundary problem. We showed, that the variation of pressure near the interface can be written in terms of the covariance of the martingales on the Schottky double.

Let us  briefly discuss a relation between the family of martingales proposed in this work, and the results of the previous work~\cite{Ale19b}. The exponential martingales introduced in Proposition~\ref{p:1m} are the correlation functions of the so-called vertex operators $V_{\alpha,\bar \alpha}(z,\bar z)$. The vertex operators are exponential functions, $V_{\alpha,\bar\alpha}(z,\bar z)=\exp(\alpha\varphi(z)+\bar \alpha\bar\varphi(\bar z))$, of the Gaussian free field $\varphi(z)$. Correlation functions of product of vertex operators can be easily computed due to the simple form of the GFF correlators, $\langle\varphi(z)\varphi(z')\rangle=-\log (z-z')$, $\langle\varphi(z)\bar\varphi(\bar z')\rangle=0$, and $\langle\bar \varphi(\bar z)\bar\varphi(\bar z')\rangle=-\log (\bar z-\bar z')$. In particular, the function $Z_N$, given by the product~\eqref{Z-def}, is the correlation function of $N$ vertex operators. We will not go into detail here, and consider the coupling of the Gaussian free field to the stochastic LG problem in future publication.

By the above short remark we wish to emphasize, that the family of martingales proposed in this paper is a subclass of martingales proposed earlier~\cite{Ale19b}. The Coulomb gas formalism, mentioned above, allows one not only to obtain the correlation functions explicitly, but also to determine critical exponents of various lattice models of statistical mechanics. Besides, this method was applied to find multifractal spectrum of harmonic measure of conformally invariant SLE curves~\cite{Gruzberg06}, because the Coulomb gas formalism captures the singular behavior of CFT correlation functions. Hence, the next step is to use the results obtained in this paper to study critical exponents associated with the stochastic LG problem.

Another motivation to study stochastic interface dynamics in the Coulomb gas framework is the Duplantier-Sheffield approach to Liouville quantum gravity~\cite{DS11}. In particular, the Gaussian free field, which naturally appears in the Coulomb gas formalism is a natural object in Liouville quantum gravity, which determines the random measure on the space of curves. We believe, that this method can be used in order to couple stochastic LG to random geometry.

\section*{Acknowledgments}
The work is supported by the Russian Science Foundation grant 19-71-30002.

\appendix

\section{The pattern formation in stochastic LG}\label{a:patterns}

The LG problem can be rewritten in terms of the Schwarz function, which provides an elegant geometrical interpretation of the interface dynamics. The Schwarz function is defined as what follows (see Ref.~\cite{DavisBook} for details). Let $\p D_t$ be a sufficiently smooth curve drawn on the plane. It can be determined by the equation $F(x,y)=0$. By replacing $(x,y)$ by the complex coordinates $(z=x+iy,\bar z=x-iy)$, and solving the equation $F((z+\bar z)/2,(z-\bar z)/2)=0$ with respect to $\bar z$ one obtains
\begin{equation}
	\bar z = S_t(z),\quad z\in \p D_t.
\end{equation}
The analytical continuation of $S_t(z)$ away from the curve is called the Schwarz function. The Schwarz function can be decomposed into a sum of two function, $S^+_t$ and $S^-_t$, that are regular in $D^+_t$ and $D_t$ respectively:
\begin{equation}\label{Spm}
	S^+_t(z)=\sum_{n=1}^\infty nt_nz^{n-1},\qquad S_t^-(z)=\frac{t_0}{z}+\sum_{n=1}^\infty v_n z^{-n-1},
\end{equation}
where $t_n$ and $v_n$ are the moments of the exterior and interior domains, $t_n=-\frac{1}{\pi n}\int_{D_t}z^{-n}d^2z$, $v_n=\frac{1}{\pi}\int_{D^+_t}z^n d^2z$. The deterministic LG problem~\eqref{z_eq},~\eqref{rho-lg} is known to be equivalent to the following equation~\cite{HowisonEJAM}:
\begin{equation}\label{SW}
	\p_t S_t(z)=2 \p_z W_t(z),
\end{equation}
where $W_t(z)= - P_t(z) + i\Psi_t(z)$ is the complex potential given by the sum of negative pressure $P_t(z)$ and the stream function $\Psi_t(z)$. Since the complex potential for the idealized deterministic LG is $W_t(z)=(Q/2\pi)\log w_t(z)$, eq.~\eqref{SW} reduces to $\p_t S^+_t(z)=0$ for $z\in D_t$. From this equation one concludes that all moments $t_n$, $n=1,2,\dotsc$, are integrals of motion.

Contrary to the deterministic problem, in the case of stochastic LG the moments, $t_k$, vary with time. The effective pressure field for stochastic LG~\eqref{P-def} attributes the local fluctuations of pressure in the vicinity of the interface to the set virtual sources in $D_t$, so that $\nabla^2 P(z,\bm{\xi})=\nu\sum_{n=1}^N \delta^{(2)}(z-\zeta_n)$, where $\zeta_n=z_t(1/\bar \xi_n)\in D_t$, $n=1,2,\dotsc,N$. Because the points $\zeta_n$ are the only singularities of $W_t(z)$ (except infinity) in $D_t$ from eq.~\eqref{SW} one obtains
\begin{equation}\label{a:S+}
	\p_t S^+_t(z)=\nu\sum_{n=1}^N\frac{1}{z-\zeta_n}.
\end{equation}
If $\zeta_n=const$ the Schwarz function has simple poles in $D_t$ with the residues linearly growing with time. From the physical point of view, this process represents LG with oil wells at the points $\zeta_n$ with the rates $\nu$. In stochastic LG the time evolution of virtual sources, $\zeta_n$, $n=1,2,\dotsc,N$, results in the formation of $N$ cuts of the Schwarz function in $D_t$. Hence, $S^+_t$, can be represented as a sum of the Cauchy type integrals,
\begin{equation}\label{a:S+Cauchy}
	S_t^+(z)=S^+_{0}(z)+2\nu\sum_{n=1}^N\int_{\gamma_n(t)}\frac{P_n(l)dl}{z-l},
\end{equation}
where the integration contours $\gamma_n(t)$, $n=1,2,\dotsc,N$, are the trajectories of the points $\zeta_n=z_t(1/\bar \xi_n)$ in $D_t$, and the Cauchy densities $P_n(l)$ for stochastic LG are determined by the velocities of these points.

It is instructive to consider stochastic LG in a discrete time framework. Let $D_I$ be domain occupied by viscous fluid at the time instant $t_I=I \delta t$. Then, the Schwarz function of its boundary, $\p D_I$, can be obtained from eq.~\eqref{a:S+} (see Ref.~\cite{Ale19a} for detail):
\begin{equation}\label{a:S+log}
	S^+_I(z)=S^+_0(z)+\sum_{n=1}^N\sum_{i=0}^Ic_{n,i}\log (z-\zeta_n(t_i)).
\end{equation}
Here, the coefficients $c_{n,i}$ in front of the logarithms are random variables, determined by velocities of the virtual source, $\zeta_n(t_I)=z_I(1/\bar \xi_n(t_I))\in D_I$, namely,
\begin{equation}\label{a:c-def}
	c_{n,i}=\nu \delta t\left(\frac{1}{\delta \zeta_n(t_i)}-\frac{1}{\delta \zeta_n(t_{i-1})}\right),
\end{equation}
where $\delta\zeta_n(t_i)=\zeta_n(t_i)-\zeta_n(t_{i-1})$ are the increments of the branch cuts of the Schwarz function, $S_t^+$, in $D_t$ during $i$th time unit. Note, that the logarithmic representation of the Schwarz function~\eqref{a:S+log} is connected to the Causchy representation~\eqref{a:S+Cauchy} upon the integration by parts.

\begin{figure}
\centering
\includegraphics[width=1\textwidth]{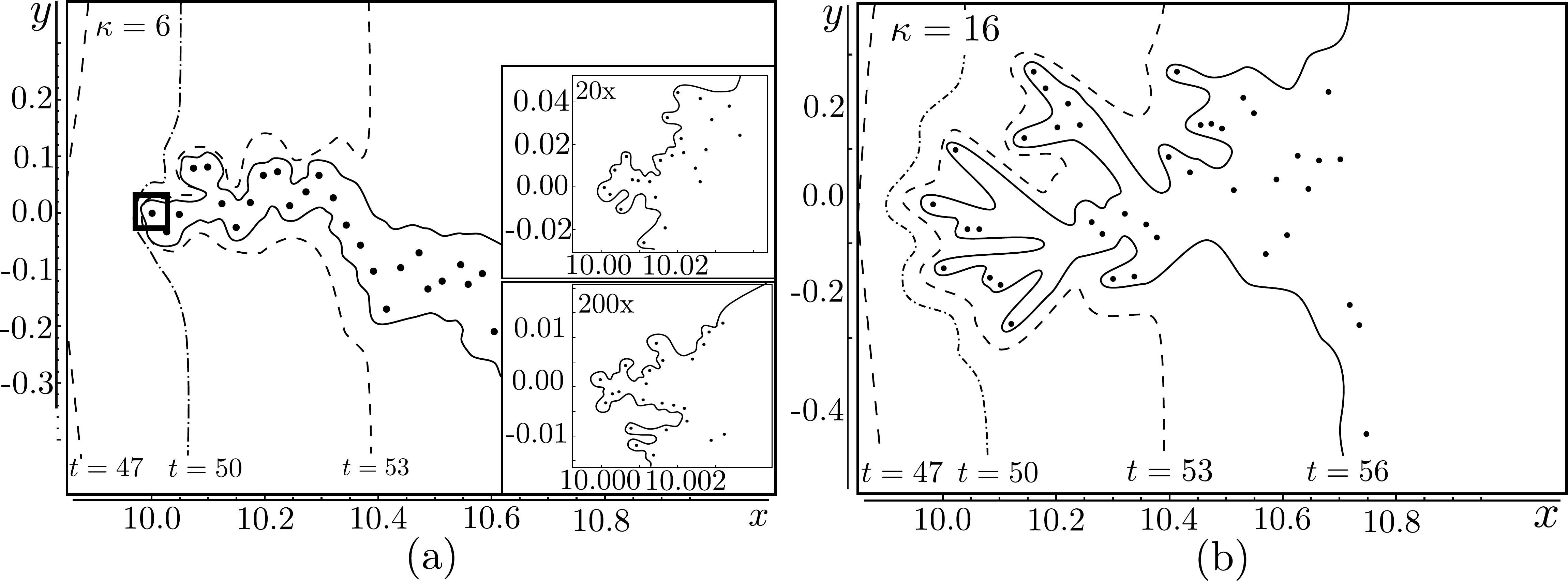}
\caption{
\label{single1}
A sequence of interfaces, $y=\Im z(e^{i\phi}, t)$ as a function of $x=\Re z(e^{i\phi},t)$, near the bottom of the fjord is plotted for two different values of the noise strength: (a) corresponds to $\kappa=6$ and (b) corresponds to $\kappa=16$. The widths of fjords are determined by the ratio $\nu/Q=0.04$. In the discrete time framework, the random process $\zeta(t)=z_t(1/\bar \xi(t))$ generates a variety of tiny fjords on the microscale, which correspond to logarithmic terms in eq.~\eqref{a:zI}. The boxed region in (a) shows the interface in the vicinity of the point $x=10$, $y=0$, magnified 20-fold and 200-fold correspondingly. Since the interface exhibits similar patterns at increasingly small scales, the interface in the bottoms of fjords has a fractal structure.}
\end{figure}

The geometrical interpretation of logarithmic singularities of the Schwarz function is straightforward~\cite{DMW98}: each logarithmic term in the right hand side of eq.~\eqref{a:S+log} corresponds to the fjords with parallel walls in the long time asymptotic. The tips of the fjords are located at the points $\{\{\zeta_n(t_i)\}_{i=0}^I\}_{n=1}^N$. The widths of the fjords are given by $\pi|c_{n,i}|$, while their orientations are specified by $\arg(c_{n,i})$. In Fig.~\ref{single1} we show a typical shape of the interface at the bottom of the fjord in the long time asymptotic. A variety of tiny fjords on the microscale form a fractal pattern in the vicinity of the bottom of deep fjords on the macroscale.

A connection of the Schwarz function to the conformal map, $S_t(z)=\bar z_t(1/w)$, where $z=z_t(w)$, implies a one-to-one correspondence between singularities of $z_t$ inside the unit disk, $\mathbb D^+$, in the $w$ plane, and singularities of $S^+_t$ in $D_t$. In particular, if $S^+_t$ has a logarithmic branch point at $\zeta_n$ with the coefficient $c_{n,i}$, then the conformal map possesses the same branch point at $\xi^i_n(t_I)=1/\bar w_I(\bar \zeta_n(t_i))$ with the coefficient $\bar c_{n,i}$. Hence, from~\eqref{a:S+log} one can immediately determine the conformal map $z_I:\mathbb D\to D_I$, namely,
\begin{equation}\label{a:zI}
	z_I(w)=r_I w+\sum_{k=1}^K \alpha_k\log (w/a_k(t_I)-1) + \sum_{n=1}^N\sum_{i=0}^I \bar c_{n,i}\log (w/\xi_n^i(t_I)-1),
\end{equation}
where $\alpha_k=const$, $k=1,2,\dotsc,K$. Here the logarithmic terms, $\alpha_k\log (w/a_k(t_I)-1)$, represent tiny deviations of the initial domain $D_0$ from a perfect circle. If the initial interface, $\p D_{0}$, is a circle, all $\alpha_k=0$. From eq.~\eqref{a:S+log} it follows, that $z_t(1/\bar a_k(t))$ are the constant of motion. Besides, logarithmic branch points of $S^{+}_I$ at the point $\zeta_n(t_i)=z_I(1/\bar \xi_n^i(t_I))$ with $i<I$ also stay constants with time for the time instants. The positions of these points, i.e., the trajectories of the processes $\zeta_n(t)$ in the $z$ plane, determine the centerlines of the fjords on the macroscale.

Contrary to the deterministic LG, where the number of singularities of the conformal map does not change, in stochastic LG the total number of singularities growth linearly with time. Let us consider a family of abelian domains, whose uniformization maps $w_t$ has rational derivatives $w_t'(z)=P(z)/Q(z)$. Let us define the degree of the abelian domain, $\deg w_t$, as follows: $\deg w_t=\max(\deg P, \deg Q)$. Let us consider the abelian domain $D_i$ of degree $M$ at time instant $t_i$. Then, a single step of stochastic LG generated by density~\eqref{rho-slg} results in the algebraic domain $D_{i+1}$ of degree not higher than $M+N$. In the continuum limit, $\delta t\to 0$, the generation of singularities of the Schwarz function can be interpreted as the evolution of the branch-cuts of the Schwarz function~\eqref{a:S+Cauchy} with time.

Now, let us briefly describe a possible asymptotic behavior of zeroes and poles of the derivative, $z'_t(w)$, of the conformal map. Although stochastic LG does not preserve the total number of poles of $z'_t(w)$, the interface evolution during $i$th growth step is equivalent to the idealized deterministic Laplacian in presence of $N+1$ oil sources: $N$ oil wells at $\zeta_n=z_t(1/\bar\xi_n(t_i))$ with rates $\nu$, and the oil sink at infinity with the rate $(Q/2\pi)+N\nu\sum \alpha_n$. Hence, an asymptotic behavior of poles and zeroes can be obtained by using the results for deterministic LG~\cite{DMW98}.

Let us consider the conformal map $z_I:\mathbb D\to D_I$, which has the form~\eqref{a:zI} with $K$ logarithmic branch-cuts inherited from the initial conditions (i.e., from singularities of $S^+_0(z)$), and $N I$ logarithmic terms generated by stochastic LG with density~\eqref{rho-slg}. Suppose that the positions of all singularities of $z_I(w)$ are known. Let us show, that stochastic LG during the next time unit, $I+1$, can be reduced to the many-body-type problem for the motion of poles of $z'_t(w)$ inside the unit disk described by a system of stochastic ordinary differential equations. Since stochastic LG is equivalent to the idealized deterministic growth with moving source, one has the following integrals of motion:
\begin{equation}\label{a:const}
	\beta_k=z_{I+1}(1/\bar a_k(t_{I+1})),\quad \gamma_{n,i}=z_{I+1}(1/\bar \xi_n^i(t_{I+1})),
\end{equation}
for $k=1,2,\dotsc,K$, $n=1,2,\dotsc, N$, and $i=1,2,\dotsc, I$. During the $(I+1)$th growth step, the conformal map develop $N$ new poles at the points, $\xi^{I+1}_n(t_{I+1})$, $n=1,2,\dotsc,N$, given by solutions to the set of coupled stochastic differential equations~\eqref{dxi}. The coefficients, $c_{n,I+1}$ in front of the logarithmic terms in $z_{I+1}(w)$ (see eq.~\eqref{a:zI}) are specified by $N$ equations~\eqref{a:A}. The conformal radius, $r_{I+1}$, can be determined by the area $\mathcal A$, enclosed by $\p D_{I+1}$. Since the total growth rate is $Q$, one has $\p_t\mathcal A=Q$, we get the equation
\begin{equation}\label{a:A}
	t_{I+1}=\frac{r_{I+1}^2-r_0^2}{2}+\sum_{m,n}\frac{\bm{\alpha}_m\bar{\bm{\alpha}}_n}{2}\log (1-\bm{a}_m\bar{\bm{a}}_n)-\sum_{k,l}\frac{\alpha_k\bar \alpha_l}{2}\log (1 - a_k(t_0)\bar a_l(t_0)),
\end{equation}
where $\bm{\alpha}_m$ and $\bm{a}_m$ are the $m$th elements from the sets $\bm{\alpha}=(\{\alpha_k\}_{k=1}^K,\{\{\bar c_{n,i}\}_{n=1}^N\}_{i=1}^{I+1})$ and $\bm{a}=(\{a_k(t_{I+1})\}_{k=1}^K,\{\{\xi^i_n(t_{I+1})\}_{n=1}^N\}_{i=1}^{I+1})$. Thus, one obtains $K+N(I+1)+N+1$ equations,~\eqref{a:const},~\eqref{dxi}, and~\eqref{a:A}, for the required parameters of the map, $r_{I+1}$, $a_k(t_{I+1})$, $\xi_n^i(t_{I+1})$, and $\bar c_{n,I+1}$, for $K=1,2,\dotsc, K$, $n=1,2,\dotsc,N$, and $i=1,2,\dotsc,I+1$. Hence, we conclude, that stochastic LG is a self-consisted problem. It can be recast in the many-body problem for the dynamics of poles of the conformal map inside the unit circle.

It is known, that the logarithmic solution~\eqref{a:zI} to idealized LG problem does not produce cusp at the interface for a wide class of initial conditions~\cite{DMW98}. Therefore, poles and zeroes of $z'_t(w)$ never hit the unit circle. A possible asymptotic behavior of poles can be obtained by considering the large time asymptotic of the right hand sides of eqs.~\eqref{a:const}. Because of eq.~\eqref{a:A}, the function $r^2_t$ growth with time, i.e., $r_t\to\infty$ as $t\to\infty$. Then, the only possibility to cancel the divergent terms in $\gamma_{n,i}$ and $\beta_k$ is to suppose that some points from the sets $\{\xi_n^i(t_I)\}$ and $\{a_k(t_I)\}$ tend to the unit circle as $t_I\to\infty$ following an exponential law. We do not give a careful analysis of the asymptotic behavior and refer the reader to Ref.~\cite{DMW98} for the similar analysis in the case of deterministic LG.

\bibliographystyle{alpha}
\bibliography{biblio}{}

\end{document}